\documentclass[a4paper,UKenglish,cleveref, autoref, thm-restate]{lipics-v2021}

\title{Loop Termination and Generalized Collatz Sequences} 

\author{Mishel Carelli}{CISPA Helmholtz Center for Information Security, Germany}{mishel.carelli@cispa.de}{https://orcid.org/0009-0000-2181-8205}{This
work was supported by the European Research Council (ERC) Grant HYPER
(No. 101055412).}

\usepackage{tikz}

\authorrunning{Mishel Carelli} 

\Copyright{Mishel Carelli} 



\category{Track B: Automata, Logic, Semantics, and Theory of Programming} 

\relatedversion{} 

\acknowledgements{I want to thank Arkadiy Aliev, Bernd Finkbeiner, and Joël Ouaknine for the useful discussions.}

\def\ZZ{\mathbb{Z}}
\def\NN{\mathbb{N}}
\def\QQ{\mathbb{Q}}
\def\RR{\mathbb{R}}
\def\lim{{\varprojlim}}

\def\->{\rightarrow}

\newcommand{\sign}{\operatorname{sign}}
\newcommand{\cl}{\operatorname{cl}}
\newcommand{\conv}{\operatorname{conv}}
\newcommand{\nonneg}{\operatorname{nonneg}}
\newcommand{\rec}{\operatorname{rec}}
\newcommand{\dist}{\operatorname{dist}}
\newcommand{\height}{\operatorname{height}}

\hideLIPIcs

\nolinenumbers

\newtheorem{question}[theorem]{Question}

\EventEditors{Sayan Bhattacharya, Danupon Nanongkai, Michael Benedikt, and Gabriele Puppis}
\EventNoEds{4}
\EventLongTitle{53rd International Colloquium on Automata, Languages, and Programming (ICALP 2026)}
\EventShortTitle{ICALP 2026}
\EventAcronym{ICALP}
\EventYear{2026}
\EventDate{July 7--10, 2026}
\EventLocation{Royal Holloway, University of London, Egham, United Kingdom}
\EventLogo{}
\SeriesVolume{374}
\ArticleNo{175}

\begin{document}

\maketitle

\begin{abstract}
Linear-constraint loops are programs whose transition relation is specified by a system of linear inequalities. The termination problem asks, given a loop, whether it admits an infinite computation. Decidability of termination remains open for linear-constraint loops over integers, rationals, and reals. We focus on loops over integers and show that they are tightly connected to generalized Collatz sequences -- integer sequences generated by maps that are linear on each residue class modulo a fixed natural number. We prove that termination of one-variable linear-constraint loops is decidable in polynomial time, provided a long-standing conjecture about generalized Collatz sequences holds. Conversely, we show that any decision procedure for one-variable loops would prove or refute specific instances of this conjecture, which remain open. Moreover, we show that if a one-variable loop has a cyclic trace, then it also has a cyclic trace of length at most two.
\end{abstract}

\section{Introduction}

Termination analysis is a long-standing topic in theoretical computer science. The Halting Problem -- deciding whether a Turing machine halts on a given input -- is undecidable, going back to Turing’s original work. Nevertheless, termination can become decidable when restricted to certain classes of programs. 

In this paper, we analyze the open problem of termination for
Single-Path Linear-Constraint Loops (SLC). This problem originates in program analysis and verification, where incomplete but efficient approaches such as ranking functions \cite{inproceedings} and transition invariants \cite{1319598} are used as building blocks in verification tools to prove termination and other liveness properties \cite{10.1007/11817963_37}, \cite{1319598}, \cite{10.1007/978-3-319-02444-8_26} and in Horn Clauses solvers \cite{10.1007/978-3-642-39799-8_61}, \cite{10.1007/978-3-030-81685-8_35}.

SLCs are loops of the form
\begin{equation} \label{SLC} \textbf{while} \ (C\boldsymbol{x}\leq d) \ \textbf{do} \ A \left( \begin{array}{c} \boldsymbol{x} \\ \boldsymbol{x}' \end{array} \right) \leq b. 
\end{equation}
Here, $\boldsymbol{x}$ is the vector of values of the program variables in the current step, and $\boldsymbol{x}'$ denotes their values in the next step.
Moreover, $A$ and $C$ are matrices with rational coefficients, and $b$ and $d$ are rational vectors. Later in the paper, we introduce an equivalent definition with the loop guard omitted. The decision problem is: given such a loop, does there exist an infinite computation with program variables taking values in $\ZZ$ from some initial state? 

The question of whether this problem is decidable was posed two decades ago \cite{10.1007/11817963_34}, and remains open to this day, not only for loops over $\ZZ$, but also over $\QQ$ and $\RR$. Recently, decidability over $\RR$ has been established in the two-variable case \cite{guilmant_et_al:LIPIcs.ICALP.2024.140}.

Decidability of termination varies across loop classes. We now discuss two additional classes of linear loops to illustrate the current landscape.

On the positive side, a prominent decidable class of linear loops called
\emph{affine} SLCs, in which the transition is deterministic and has the form
\[\mathbf{while}\ (C\boldsymbol{x} \le d)\ \mathbf{do}\ \boldsymbol{x}^{\prime} := A\boldsymbol{x} + b.\]
Termination of affine SLCs has been proved decidable over $\RR$ \cite{Tiwari04:CAV} -- via spectral analysis of $A$ -- over $\QQ$ \cite{10.1007/11817963_34}, by symbolic reasoning about matrix powers through eigenvalues, and over $\ZZ$ \cite{hosseini_et_al:LIPIcs.ICALP.2019.118}, using different tools from geometry of numbers, such as the Khachiyan–Porkolab result on deciding whether a convex semi-algebraic set contains an integer point.

On the negative side, another relevant class of linear loops called \emph{Multi-Path Linear-Constraint Loops}, which allows conditional choice among several linear updates:
\[
\mathbf{while}\ (\text{true})\ \mathbf{do}\
\boldsymbol{x}^{\prime} :=
\begin{cases}
A_1\boldsymbol{x}+b_1 & \text{if } C_1\boldsymbol{x}\le d_1,\\
\vdots\\
A_k\boldsymbol{x}+b_k & \text{if } C_k\boldsymbol{x}\le d_k.
\end{cases}
\]
Termination of such multi-path loops is undecidable over $\RR$ and $\QQ$ \cite{10.1016/S0304-3975(00)00399-6} and over $\ZZ$ \cite{benamram:LIPIcs.STACS.2013.514}, even for $k=2$, via encodings of counter machines.

In \cite{10.1007/11817963_34}, Braverman asked: ``How much nondeterminism can be introduced in a linear loop with no initial conditions before termination becomes undecidable?’’ Motivated by this question and in light of the subsequent results, SLCs occupy an intriguing middle ground between deterministic, decidable affine SLCs and highly expressive, undecidable multi-path loops. Because nondeterminism in SLCs is convex, it does not allow for modeling the unstructured nondeterminism of counter machines, at least not in a straightforward way.


In~\cite{10.1145/2400676.2400679}, the authors show that allowing a single irrational coefficient makes SLC termination undecidable over $\ZZ$, and they describe the SLC termination problem as ``the most intriguing'' among loop-termination problems.

For a complete explanation of these results and further references, we refer the reader to the recent survey on linear loop termination \cite{benamram2025terminationanalysislinearconstraintprograms}.


In this work, we study the SLC termination problem for the case of one variable. While this restriction may seem severe, we show that this problem is deeply non-trivial even for one-variable loops, connecting it with a long-standing number theoretic conjecture about \emph{generalized Collatz sequences}.

The study of \emph{Generalized Collatz sequences} is a long-standing topic in number theory \cite{Matthews1984}, \cite{Matthews1985AMA}, \cite{Matthews1992SomeBM}, \cite{Mller1978berHV}, \cite{Leigh1986AMP}, \cite{matthews2010generalized}. These integer sequences are generated by generalized Collatz mappings of the form
$$T(x) = 
\frac{m_i x-r_i}{d}, \ \ \ \ \text{if }  x\equiv i\text{ (mod } d)$$
The \textit{Uniform Distribution conjecture} \cite{Matthews1984} asserts that every unbounded generalized Collatz sequence is uniformly distributed modulo $d^{\alpha}$ for every $\alpha\in\NN$. We formulate a weaker statement, the \emph{Reachability Conjecture}, which posits that every sequence generated by generalized Collatz mappings of a certain type called weak Collatz mappings, reaches a certain residue class modulo $d$ at least once. We prove the Reachability Conjecture for $d=2$.

Our main result shows that, if the Reachability Conjecture holds, then termination of one-variable SLCs is decidable in polynomial time. 

For the other direction, we show that any algorithm deciding termination of one-variable SLCs would, given a weak Collatz mapping, determine whether the Reachability Conjecture holds for all sequences generated by that mapping. All cases with modulus $d>2$ remain open. This suggests that further progress on SLC decidability will require a deeper understanding of generalized Collatz sequences.

Let us give the high-level explanation of the results of this paper and its structure. 

\textbf{Section~\ref{preliminaries}} introduces the basic definitions and notation. In particular, we define the SLC termination problem and split it into two subproblems: deciding the existence of \emph{cycles} (finite traces from some state to itself), and deciding the existence of \emph{self-avoiding traces} (infinite traces that never repeat a state).

\textbf{Section~\ref{cycles}} studies cycles. We show that deciding the existence of a cycle of a bounded length can be reduced to an instance of \emph{Integer Linear Programming} problem. The main result of this section is that for a one-variable loop, the existence of any cycle implies the existence of a cycle of length at most two; hence cycle-detection problem is decidable for one-variable loops.

\textbf{Section~\ref{collatz}} introduces generalized Collatz sequences, and states the Uniform Distribution and Reachability Conjectures. We prove the Reachability conjecture for the case $d=2$, and we prove that a one-variable SLC termination-checking algorithm, if it existed, could verify specific open instances of this conjecture.

\textbf{Section~\ref{unbounded}} analyzes self-avoiding traces and shows that, assuming the Reachability Conjecture, termination of one-variable SLCs is decidable in polynomial time. The proof is geometric in nature. We view the transition relation as a polyhedron in $\RR^2$ and analyze its \emph{Minkowski-Weyl decomposition}. Intuitively, we identify where the transition polyhedron ''points to''. Simplifying, the self-avoiding trace exists if the polyhedron points in the direction where the absolute value of the variable grows. However, several corner cases arise, in one of which, the trace of SLC forms a weak Collatz sequence, which avoids a certain residue class modulo $d$. If the Reachability conjecture holds, we can conclude that this kind of execution needs to terminate. Otherwise, the presented approach results in a semi-algorithm, which does not give an answer in these specific cases.

\textbf{Section~\ref{proofs}} presents the formal proofs of the results from Sections~\ref{cycles}, \ref{collatz}, and~\ref{unbounded}.

\textbf{Section~\ref{conclusion}} summarizes the results.

\section{Preliminaries}\label{preliminaries}

\subsection{Notation}
We denote the sets of natural numbers, integers, rationals, and reals by $\NN, \ZZ, \QQ$, and $\RR$, respectively. The sets of positive, negative, and nonnegative real numbers are denoted by $\RR_{+}$, $\RR_{-}$, and $\RR_{\ge 0}$.

Given $S \subseteq \RR^{n}$, we write $S_{\ZZ} := S \cap \ZZ^{n}$ for its set of integer points.

Given $T \subseteq \RR^{2n}$ and $z \in \RR^{n}$, define the slice
$T_{z} := \{ x_{2} \in \RR^{n} \mid (z, x_{2}) \in T \}$.

Vectors $v_{1}, v_{2} \in \RR^{n}$ are \emph{collinear} if there exists $r \in \RR$ such that $v_{1} = r v_{2}$; we write $v_{1} \parallel v_{2}$. If $v_{1}$ and $v_{2}$ are not collinear, we write $v_{1} \nparallel v_{2}$.

For $p, q \in \ZZ$, $\gcd(p,q)$ denotes the greatest natural number that divides both $p$ and $q$.

For $x\in \RR$ define

\[
\sign(x):=
\begin{cases}
  1 & \text{if } x>0,\\
 -1 & \text{if } x<0,\\
  0 & \text{if } x=0.
\end{cases}
\]

\subsection{Convex sets and the Minkowski–Weyl theorem}

Given $S_{1}, S_{2} \subseteq \RR^{n}$, the \emph{Minkowski sum} is
\[S_{1} + S_{2} := \{\, y + z \mid y \in S_{1},\, z \in S_{2} \}.\]

A set $K \subseteq \RR^{n}$ is \emph{convex} if for all $x,y \in K$ and all $\lambda \in [0,1]$,
we have $\lambda x + (1-\lambda) y \in K$.

The \emph{topological closure} of $S \subseteq \RR^{n}$ is denoted $\cl(S)$ (the set of limits of the sequences of elements of $S$). A set $S$ is \emph{closed} if $S = \cl(S)$.

Intersections of families of closed sets are closed, and intersections of families of convex sets are convex.

\begin{definition}
Let $S \subseteq \RR^n$.
\begin{enumerate}
    \item A convex hull of $S$ is defined by
    \[\conv(S) = \{ \sum_{i=1}^k \alpha_i v_i  \mid k\in \NN, \ \alpha_i \in [0;1],  \ v_i\in S , \ \sum_{i=1}^k \alpha_i = 1 \} \]
    \item A nonnegative (conic) hull is defined by
    \[ \nonneg(S) =  \{ \sum_{i=1}^k \alpha_i v_i  \mid k\in \NN,\ \alpha_i \in \RR_{\geq 0},  \ v_i\in S \}\]
\end{enumerate}
\end{definition}

A \emph{convex cone} is a set $C \subseteq \RR^{n}$ such that for every $c_{1}, c_{2} \in C$ and every $\lambda \in \RR_{\ge 0}$, we have $\lambda c_{1} \in C$ and $c_{1} + c_{2} \in C$. Every convex cone is a convex set.

\begin{definition}
Let $K \subseteq \RR^{n}$ be a nonempty convex set. The \emph{recession cone} of $K$ is
\[
\rec(K) := \{\, v \in \RR^{n} \mid K + \RR_{\ge 0} v \subseteq K \,\}.
\]
\end{definition}

For every nonempty convex $K$, $\rec(K)$ is a closed convex cone. A cone $K$ is pointed if $K\cap(-K) = 0$.

The following classical result appears in \S8.9 of \cite{10.5555/17634}.

\begin{theorem}[Minkowski–Weyl Theorem]\label{Minkowski-Weyl}
Let $P \subseteq \RR^{n}$. The following are equivalent:
\begin{enumerate}
\item $P$ is a polyhedron, i.e., $P = \{ x \in \RR^{n} \mid A x \le b \}$ for some $A \in \RR^{k \times n}$ and $b \in \RR^{k}$.
\item There exist $w_1,\dots w_l \in P$ and $v_1,\dots, v_m \in \RR^n$, such that
        \[ P = \conv(w_1,\dots,w_l) + \rec(P) =  \conv(w_1,\dots,w_l) + \nonneg(v_1,\dots,v_m) \]
\end{enumerate}
Moreover, if $P$ is given by linear inequalities with rational coefficients, then $w_{i}$ and $v_{i}$ can be chosen with rational entries.
\end{theorem}

Note that any polyhedron $P \subseteq \RR^{n}$ is closed and convex, because every linear inequality defines a closed half-space in $\RR^{n}$ and $P$ is an intersection of finitely many half-spaces.

The following result trivially follows from basic results in convex geometry, but we present the proof of this exact formulation for the sake of completeness.

\begin{proposition}\label{classification}
The following statements hold for closed convex cones $C \subseteq \RR^{2}$.
\begin{enumerate}
\item $C$ can be represented as $\nonneg(v_{1},\dots,v_{m})$ for some $v_{i} \in \RR^{2}$ with $m \le 3$.
\item If $v_{1}, v_{2} \in \RR^{2}$ and $v_{1} \nparallel v_{2}$, then $\nonneg(v_{1},v_{2})$ contains no line in $\RR^{2}$.
\item If $C$ cannot be generated by two vectors, then $C$ is either a half-plane or the whole plane.
\end{enumerate}
\end{proposition}

\begin{proof}
By \S8.9 of \cite{10.5555/17634}, $C$ decomposes as $C = C^{\prime} + \operatorname{lin}(C)$, where $\operatorname{lin}(C) := C \cap (-C)$ is the lineality space and $C^{\prime}$ is a pointed cone.

If $\operatorname{lin}(C) = \RR^{2}$, then $C = \RR^{2} = \nonneg((0,1),(1,0),(-1,-1))$.

Assume $\operatorname{lin}(C)$ is a line, say $\nonneg(v,-v)$, and $C \ne \RR^{2}$. Then $\RR^{2} = H_{1} \cup H_{2} \cup \operatorname{lin}(C)$, where $H_{1}$ and $H_{2}$ are the open half-planes with boundary $\operatorname{lin}(C)$. If $C^{\prime} \cap H_{i} \ne \emptyset$ for some $i \in {1,2}$, then $H_{i} \subseteq C$. Hence $C$ is either the line $\nonneg(v,-v)$ or, if there exists $v^{\prime} \in C^{\prime} \setminus \operatorname{lin}(C)$, a half-plane $\nonneg(v,-v,v^{\prime})$.

Finally, if $\operatorname{lin}(C) = {0}$, then $C = C^{\prime}$ is pointed and therefore generated by two vectors (see, e.g., Example 2.39 in \cite{telen2022introductiontoricgeometry}); in this case, it cannot contain a line by definition of pointedness.
\end{proof}

The next statement is Corollary 8.3.3 in \cite{rockafellar2015convex}.

\begin{proposition}[Recession cone of an intersection]\label{infcap}
If $P_{1}$ and $P_{2}$ are closed convex sets with $P_{1} \cap P_{2} \ne \emptyset$, then $\rec(P_{1} \cap P_{2}) = \rec(P_{1}) \cap \rec(P_{2})$.
\end{proposition}

\subsection{Loops}

We consider Single-Path Linear-Constraint Loops (\textit{SLC}). Throughout, SLCs are over the integers.

The vector of variables is $\boldsymbol{x} = (x_{1},\dots,x_{n})^{T}$. A \emph{state} is a valuation over integers, i.e., a vector in $\ZZ^{n}$. The vector of variables at the next step is $\boldsymbol{x}^{\prime} = (x_{1}^{\prime},\dots,x_{n}^{\prime})^{T}$. An SLC is specified by a system of linear inequalities
\begin{equation} \label{polyhedron}
        A \left( \begin{array}{c} \boldsymbol{x} \\ \boldsymbol{x}^{\prime} \end{array} \right) \leq b,
    \end{equation}
where $A \in \ZZ^{k \times 2n}$ and $b \in \ZZ^{k}$.

\begin{remark}
This definition differs from \eqref{SLC} in the Introduction because it omits an explicit loop guard. The two forms are expressively equivalent, since the guard can be encoded in the transition relation.
\end{remark}

The \emph{transition polyhedron} is the set $\mathcal{R} \subseteq \RR^{2n}$ of vectors satisfying \eqref{polyhedron}.

\begin{remark}
Since we study integer-valued loops, the actual transitions between states are elements of $\mathcal{R}_{\ZZ}$. We nevertheless work over the reals when convenient for proofs.
\end{remark}

We identify an SLC with its transition polyhedron $\mathcal{R}\subseteq \RR^{2n}$, and denote it by $\mathcal{R}$. We denote the size of the binary encoding of $A$ and $b$ as $|\mathcal{R}|$.
 
A \emph{trace} is a sequence of states $s_1, s_{2}, \dots$ such that $(s_{i}, s_{i+1}) \in \mathcal{R}_{\ZZ}$ for all $i\ge 1$.

The \emph{SLC termination problem} is to decide, for a given SLC, whether it has an infinite trace.

\subsection{Cycles and Self-avoiding Traces}
We divide the termination problem into two problems for two different types of infinite traces.

For an SLC, a \emph{cycle of length $k$} is a finite trace $s_{1},\dots,s_{k},s_{1}$. The existence of a cycle implies the existence of an infinite trace. The converse need not hold:

\begin{example}
The SLC $x^{\prime} = x + 1$ has no cycle but admits an infinite trace.
\end{example}

An infinite trace $s_{1}, s_{2}, \dots$ is \emph{self-avoiding} if $s_{i} \ne s_{j}$ for all $i \ne j$.

\begin{observation}
An SLC has an infinite trace if and only if it has a cycle or a self-avoiding trace.
\end{observation}

In the remainder of the paper, we analyze the problems of deciding the existence of a cycle and the existence of a self-avoiding trace separately for a given SLC.

\section{Cycles}\label{cycles}

\subsection{Bound on a cycle length}

Although cycles are finite objects (unlike self-avoiding traces), deciding their existence for an SLC remains nontrivial, since a cycle can have arbitrary length. By contrast, deciding the existence of a cycle whose length is bounded is decidable.

\begin{proposition}\label{bounded_ILP}
Given an SLC $\mathcal{R}$ and $M \in \NN$, deciding whether $\mathcal{R}$ has a cycle of length at most $M$ is in $\mathsf{NP}$ with respect to $|\mathcal{R}| \cdot M$.
\end{proposition}

\begin{proof}
It is known that the feasibility problem for Integer Linear Programming (ILP) -- deciding whether a rational polyhedron contains an integer point -- is $\mathsf{NP}$-complete \cite{10.1145/322276.322287}. We give a polynomial-time reduction from our problem to ILP.

Given $\mathcal{R}$ and $M$, for each $m \in \{1,\dots, M\}$ we construct an ILP instance encoding the existence of a cycle of length \emph{exactly} $m$:
$$P_m = \bigl\{ (\boldsymbol{x}_1,\dots,\boldsymbol{x}_m )\mid (\boldsymbol{x}_m,\boldsymbol{x}_1) \in \mathcal{R} \wedge \forall i<m: (\boldsymbol{x}_{i}, \boldsymbol{x}_{i + 1}) \in \mathcal{R} \bigr\}.$$
Then $\mathcal{R}$ has a cycle of length at most $M$ iff one of the polyhedra $P_m$ contains an integer point. Each ILP instance has size $ \mathcal{O}(|\mathcal{R}| \cdot M)$.
\end{proof}

Proposition~26 of \cite{guilmant_et_al:LIPIcs.ICALP.2024.140} shows that for Linear-Constraint Loops over $\QQ$ or $\RR$, the existence of a bounded infinite sequence is equivalent to the existence of a cycle of length one. The next proposition shows that this equivalence fails over $\ZZ$.

\begin{proposition}\label{lower_cycle}
For every $n$, there exists an SLC over $n$ variables that has a cycle of length $2^{n}$ and no cycles of smaller length.
\end{proposition}

\begin{proof}
Fix $n$ and consider the set of states $S := \{0,1\}^{n}$. Take any cycle $C$ on $S$, so that each $s \in S$ appears exactly once; $|C| = 2^{n}$. Let $T$ be the set of transitions of $C$, so $T \subseteq \{0,1\}^{2n}$. Let $\mathcal{R}:= \conv(T)$, which as a convex hull of a finite set of integer points, can be defined as an SLC. We claim that $\mathcal{R}_{\ZZ} = T$, i.e., $\conv(T)$ has no integer points other than its vertices in $T$. This implies that the minimal cycle length of $\conv(T)$ is $2^n$.

Suppose, to the contrary, that there exists $p \in \mathcal{R}_\ZZ \setminus T$. Enumerate $T = \{t^{1},\dots,t^{2^{n}}\}$ and write $p = \sum_{i=1}^{2^{n}} \alpha_{i} t^{i}$ with $\alpha_{i} \in [0,1]$ and $\sum_{i} \alpha_{i} = 1$. Without loss of generality, assume $\alpha_{1} > 0$ and choose a coordinate $j$ such that $t^1_j\ne p_j$ (such a $j$ exists since $ t^{1}\ne p$). Because $p_{j} = \sum_{i=1}^{2^n} \alpha_{i} t^{i}_j$ is an integer and each $t^{i}_{j} \in \{ 0,1 \}$, the only way the convex combination can be an integer is if all $t^{i}_{j}$ with $\alpha_{i} > 0$ are equal. This contradicts $t^{1}_{j} \ne p_{j}$. Hence, no such $p$ exists.
\end{proof}

Since there is no universal bound on the minimal cycle length, it is natural to ask whether a bound depending on the number of variables exists.

\begin{question}\label{loops_cycle}
For every $n \in \NN$, does there exist a constant $M(n)$ such that every loop over $n$ variables that has a cycle also has a cycle of length at most $M(n)$?
\end{question}


If the answer to Question~\ref{loops_cycle} is positive, then cycle detection reduces to checking for cycles of length at most $M(n)$. If $M(n)$ is computable, this yields decidability of cycle detection.

There exist convex sets that are not polyhedra (e.g., convex sets defined by quadratic inequalities). We now show that a positive answer to Question~\ref{loops_cycle} extends from polyhedral relations to arbitrary convex relations.

\begin{proposition}
Suppose that every SLC $\mathcal{R} \subseteq \RR^{2n}$ with a cycle also has a cycle of length at most $M(n)$. Then any convex relation $\mathcal{T} \subseteq \RR^{2n}$ with a cycle also has a cycle of length at most $M(n)$.
\end{proposition}

\begin{proof}
If $\mathcal{T}$ has a cycle represented by the set of transitions $\{t_{1},\dots,t_{m}\} \subseteq \mathcal{T}$, then $\conv(t_{1},\dots,t_{m}) \subseteq \mathcal{T}$ by convexity. The polyhedron $\conv(t_{1},\dots,t_{m})$ has a cycle and therefore, by hypothesis, has one of length at most $M(n)$. This cycle is contained in $\mathcal{T}$.
\end{proof}

\subsection{One variable case}

The following Theorem shows that the answer to \cref{loops_cycle} is positive for $n=1$ with $M(1)=2$.

\begin{theorem}\label{onevarcycle}
    If an SLC $\mathcal{R}\subseteq \RR^2$ has a cycle, then it has a cycle of length at most $2$.
\end{theorem}

\begin{proof}[Sketch of proof]
        
    Let us informally explain the idea of the proof with an illustration from \Cref{fig:cycle}. Consider the set of transitions of a cycle (the black points). The center of this set lies on the diagonal (the green point). If it has integer coordinates, then we obtain a length-one cycle. Otherwise, it lies on a unit square in the integer lattice. Two vertices of this unit square lie on the diagonal representing length-one cycles (the blue points). The other two vertices lie on the anti-diagonal, representing a length-two cycle (the red points). We then consider several cases, branching on the signs of different values, such as the largest value among the states of the cycle. In some cases, we show that $\mathcal{R}$ contains points on the diagonal below or above the blue points, yielding the existence of the length-one cycles. In the remaining cases, we show $\mathcal{R}$ contains points on the anti-diagonal on both sides of the green point, implying the existence of a length-two cycle. The formal argument is presented in \Cref{proofs}.

\end{proof}

\begin{figure}[t]
\centering

\begin{subfigure}[t]{0.48\textwidth}
  \centering
  \begin{tikzpicture}[scale=0.8]
    \draw[step=1, very thin] (-3,-3) grid (3,3);

    \draw[->] (-3.4,0) -- (3.4,0) node[right] {$x$};
    \draw[->] (0,-3.4) -- (0,3.4) node[above] {$x'$};

    \foreach \t in {-3,-2,-1,1,2,3} {
      \draw (\t,0.08) -- (\t,-0.08) node[below] {\small $\t$};
      \draw (0.08,\t) -- (-0.08,\t) node[left] {\small $\t$};
    }
    \node[below left] at (0,0) {\small $0$};

    \coordinate (A) at (1,2);
    \coordinate (B) at (2,0);
    \coordinate (C) at (0,-2);
    \coordinate (D) at (-2,1);
    \draw[thick] (A) -- (B) -- (C) -- (D) -- cycle;

    \fill[green] (0.25,0.25) circle (2.5pt);
    \fill[blue]  (0,0) circle (2.5pt);
    \fill[blue]  (1,1) circle (2.5pt);
    \fill[purple] (1,0) circle (2.5pt);
    \fill[purple] (0,1) circle (2.5pt);
  \end{tikzpicture}
  \caption{Cycle intuition.}
  \label{fig:cycle}
\end{subfigure}
\hfill
\begin{subfigure}[t]{0.48\textwidth}
  \centering
  \begin{tikzpicture}[scale=0.85]
    \def\xmin{0}
    \def\xstart{3}
    \def\xmax{7}
    \def\ymin{0}
    \def\ymax{10}

    \draw[step=1, very thin] (\xmin,\ymin) grid (\xmax,\ymax);

    \draw[->] (\xmin-0.3,0) -- (\xmax+0.5,0) node[right] {$x$};
    \draw[->] (0,\ymin-0.3) -- (0,\ymax+0.5) node[above] {$x'$};

    \foreach \t in {\xmin,...,\xmax} {
      \draw (\t,0.10) -- (\t,-0.10) node[below] {\small $\t$};
    }

    \foreach \t in {1,...,\ymax} {
      \draw (0.10,\t) -- (-0.10,\t) node[left] {\small $\t$};
    }

    \draw[thick,->] (\xstart,{(4*\xstart-2)/3}) -- (\xmax,{(4*\xmax-2)/3})
      node[above right] {};
    \draw[thick,->] (\xstart,{(4*\xstart-1)/3}) -- (\xmax,{(4*\xmax-1)/3})
      node[above right] {};

    \fill[red, opacity=0.25]
      (\xstart,{(4*\xstart-2)/3}) --
      (\xmax,{(4*\xmax-2)/3}) --
      (\xmax,{(4*\xmax-1)/3}) --
      (\xstart,{(4*\xstart-1)/3}) -- cycle;
  \end{tikzpicture}
  \caption{SLC $\mathcal{R}^+$.}
  \label{fig:Rplus}
\end{subfigure}

\caption{}
\label{fig:twofigs}



\end{figure}

\section{Generalized Collatz sequences} \label{collatz}
The classical Collatz problem remains open. A key difficulty is its residue-class–dependent dynamics modulo~$2$.
Several generalizations have been studied, replacing the modulus~$2$ by other moduli and considering related piecewise-linear maps~\cite{lagarias20213x1problemoverview}.
In particular, the framework proposed by Möller~\cite{Mller1978berHV} has been investigated both theoretically and empirically~\cite{matthews2010generalized}.
In this paper, we focus on the Uniform Distribution Conjecture of Matthews and Watts~\cite{Matthews1984}, formulated for Möller's generalization.

Let $d$ be a positive integer with $d\geq 2$. Let $m_0,\dots, m_{d-1}$ be nonzero integers such that $\gcd(m_i,d) = 1$  for every $i$. Let $R = \{ 
r_0,\dots,r_{d-1} \}$ be a set of integers with the property that for each $  i,$ $m_i i \equiv r_i (mod \ d)$. Then the generalized Collatz mapping $T: \ZZ \to \ZZ$ is given by
$$T(x) := 
\frac{m_i x-r_i}{d}, \ \ \ \ \text{if }  x\equiv i\text{ (mod } d)$$
For $n\in \ZZ$ we define the sequence $\{T^k(n)\}_{k\in \NN}$ recursively by
$$T^k(n) := \begin{cases} n, \ \ \ \ \ \ \ \ \ \ \ \ \  \ \ \ \ \text{if } k=0, \\
T(T^{k-1}(n)), \ \ \ \ \text{if } k>0
\end{cases}$$

A sequence $\{s_k\}_{k\in \NN}$ is called \emph{unbounded} if for every $n\in \NN$ there exists $k\in \NN$ such that $|s_k| > n$.

We recall a longstanding number-theoretic conjecture, first stated by Matthews and Watts \cite{Matthews1984} and further discussed, with empirical support, in \cite{matthews2010generalized}.
\begin{conjecture} [Uniform Distribution Conjecture] \label{strongconj}
   If the sequence $\{T^k(n)\}_{k\in \NN}$ is unbounded, it is uniformly distributed modulo $d^\alpha$ for each $\alpha\geq 1$, i.e. 
   $$\underset{N\to\infty}{\mathrm{lim}}\; \frac{1}{N} card\{ k\leq N \mid T^k(n) \equiv i \ (mod \ d^\alpha )  \} = \frac{1}{d^\alpha}$$
   for $i = 0,\dots, d^\alpha-1$.
\end{conjecture}

In this paper, we state a weaker version of this conjecture. First, let us define a more restricted class of generalized Collatz mappings.

Let $d$ be a positive integer $d\geq 2$, let $m$ be a nonzero integer with $\gcd (d,m) = 1$, and let $a$ be an integer. Then a \textit{weak Collatz mapping} $T:\ZZ \-> \ZZ$ is given by
$$T(x) := 
\frac{m x- (a+i)}{d}, \ \ \ \ \text{where } 0\le i < d, \text{ with }  mx\equiv a+i \ (\text{mod } d).$$ 

\begin{conjecture}[Reachability Conjecture]\label{simpleconj} Let $T$ be a weak Collatz mapping. Suppose for some $n\in \ZZ$ the sequence $\{T^k(n)\}_{k\in \NN}$  is unbounded. Then there exists $k \ge 0$ such that $mT^k(n) \equiv a \ (mod  \ d)$.
\end{conjecture}

The proofs of the next two propositions are presented in \Cref{proofs}.

The following proposition implies that Conjecture \ref{simpleconj} holds for mappings with $d=2$.

\begin{proposition}
    \label{collatz_lemma}
    Let $T$ be a generalized Collatz mapping with $d>1$. Suppose for some $n\in \ZZ$ the sequence $\{T^k(n)\}_{k\in \NN}$  is unbounded. Then it visits at least two residue classes modulo $d$.
\end{proposition}

The next proposition links weak Collatz sequences to self-avoiding traces.

\begin{proposition}\label{reduction}
There exists a linear-time algorithm that takes as input a weak Collatz mapping
$T(x)=\frac{mx-(a+i)}{d},$ with $m>d$, and outputs two one-variable SLCs, $\mathcal{R}_+$ and $\mathcal{R}_-$, such that:
\begin{enumerate}
    \item $\mathcal{R}_+$ has an infinite trace if and only if there exists
    $n\in \mathbb{Z}$ such that, for every $k\ge 0$,
    \[
    mT^k(n)\not\equiv a \pmod d
    \quad \text{and} \quad
    T^k(n)>0.
    \]

    \item $\mathcal{R}_-$ has an infinite trace if and only if there exists
    $n\in \mathbb{Z}$ such that, for every $k\ge 0$,
    \[
    mT^k(n)\not\equiv a \pmod d
    \quad \text{and} \quad
    T^k(n)<0.
    \]
\end{enumerate}
\end{proposition}



Proposition \ref{reduction} implies that any decidability procedure for one-variable SLC termination, if it existed, could, given a weak Collatz mapping, determine whether the Reachability Conjecture holds for this particular mapping. This suggests that progress on decidability requires progress in understanding the behavior of weak Collatz sequences. We conclude with an example of a mapping for which the Reachability Conjecture remains, to our knowledge, open, and we are unaware of any results that imply it.

\begin{example}
Consider the weak Collatz mapping
   $$T(x) = 
\frac{4 x- i}{3}, \ \ \ \ \text{if }  4x\equiv i \ (\text{mod } 3),$$
which can be rewritten as \[T(x) = \lfloor \frac{4x}{3}\rfloor. \]
For this particular $T$, it remains open whether it is true that for every $n\ge 3$ there exists $k \in \NN$ such that $3 \mid T^{k}(n)$.

Proposition~\ref{reduction} implies that this question is equivalent to the termination question for the SLC defined by
$$\mathcal{R}^+:= \{ (x,x')\in \RR^2\mid 4x-2 \le 3x' \le 4x-1 \text{ and } x\ge3 \}.$$
On the \Cref{fig:Rplus}, $\mathcal{R}^+$ is depicted as a red area, bounded by the two lines $x' = \frac{4x-2}{3}$ and $x'=\frac{4x-1}{3}$, starting from $x=3$, since for $3\leq x$, $x<\lfloor\frac{4x}{3}\rfloor $. 
\end{example}

\section{Self-avoiding Traces} \label{unbounded}

The following theorem is the main result of the paper.

\begin{theorem} \label{main}
    If the Reachability Conjecture holds, then the termination of SLCs over one variable is decidable in polynomial time.
\end{theorem}

This section is devoted to proving Theorem~\ref{main}. The treatment of cycles is given in \Cref{cycles}. Now we turn our attention to self-avoiding traces. To determine the existence of self-avoiding traces, we make a case analysis, based on the geometric placement of $\mathcal{R}$ in $\RR^2$,
split into four Lemmas~\ref{two_diff}, \ref{one}, \ref{two_col}, and \ref{simple_case}. Before stating these lemmas, we introduce an auxiliary SLC that splits $\RR^2$ into different parts.

Define the following sets. 
\[I_+ = \{ (x_1,x_2)\in \RR^2 \mid 0<x_1<x_2 \}, \ \ \ \ \ \ \ \ \ \ \ I_- = \{ (x_1,x_2)\in \RR^2 \mid x_2 < x_1 < 0\},\]
\[\Delta_+ = \{ (x_1,x_2)\in \RR^2 \mid 0<x_1=x_2\},\ \
 \ \  \ \ \ \ \ \ \Delta_- = \{ (x_1,x_2)\in \RR^2 \mid x_1 = x_2 < 0\},\]
\[ \Delta = \Delta_+\cup \Delta_-,\  \ \ \ \ \ \ \ \ I = I_+ \cup I_- .\]

Note that the topological closure of $I_+$, $I_-$ and $\Delta$ are
\[\cl(I_+) = I_+ \cup \Delta_+ \cup \{ (0,x_2)\in\RR^2 \mid x_2\geq 0\},\] \[ \cl(I_-) = I_- \cup \Delta_- \cup \{ (0,x_2)\in\RR^2 \mid x_2\leq 0\},\]
\[ \cl(\Delta) = \Delta \cup \{ (0,0)\}.\]

Recall that any SLC $\mathcal{R}\subseteq \RR^2$ can, by Theorem \ref{Minkowski-Weyl}, be represented as 
\[\mathcal{R} = \conv(w_1,\dots w_l) + \rec(\mathcal{R}) =\conv(w_1,\dots,w_l) + \nonneg(v_1,\dots,v_m),\] with $m\leq 3$. Moreover, since $\mathcal{R}$ is defined by the set of integer linear inequalities, we may assume that all $w_i$ and $v_i$ have rational coefficients.

The following case distinction, presented in four lemmas, is aimed at characterizing when a given SLC admits a self-avoiding trace, and is based on the number of generators of $\rec(\mathcal{R})$ and on the intersection of $\rec(\mathcal{R})$ with $I_+, I_-,\Delta_+$ and $\Delta_-$. Each lemma treats SLCs whose recession cone has a specific number of generators.

\begin{lemma}\label{two_diff}
Let $\mathcal{R}$ be an SLC. Suppose $\rec(\mathcal{R}) = \nonneg( v_1,v_2)$ and $v_1\nparallel v_2$. Then:
\begin{enumerate}

    \item \label{two_dif_1} If $\rec(\mathcal{R}) \cap I \neq \emptyset$, then the loop has a self-avoiding trace.
    
    \item \label{two_dif_2} If $\rec(\mathcal{R}) \cap (I \cup \Delta) = \emptyset$, then the loop does not have a self-avoiding trace.
    
    \item \label{two_dif_3}  If $\rec(\mathcal{R}) \cap \Delta_+ \neq \emptyset$ and  $ \mathcal{R}_{\ZZ}\cap I_+ \neq \emptyset$, then the loop has a self-avoiding trace.

    \item \label{two_dif_4} If $\rec(\mathcal{R}) \cap (I \cup \Delta) \subseteq \Delta_-$ and $\mathcal{R}_{\ZZ}\cap I_- = \emptyset$, then the loop does not have a self-avoiding trace.
    
    \item \label{two_dif_5} If $\rec(\mathcal{R}) \cap \Delta_- \neq \emptyset$ and  $ \mathcal{R}_{\ZZ}\cap I_- \neq \emptyset$, then the loop has a self-avoiding trace.

    \item \label{two_dif_6} If $\rec(\mathcal{R}) \cap  (I\cup \Delta) \subseteq \Delta_+$ and $\mathcal{R}_{\ZZ}\cap I_+ = \emptyset$, then the loop does not have a self-avoiding trace.
\end{enumerate}
     
\end{lemma}

\begin{proof}[Sketch of proof]
This lemma handles the case where $\rec(\mathcal{R})$ is pointed and generated by two noncollinear vectors. In this regime, $\mathcal{R}$ is “large” in the sense that the number of integer points in its vertical slices, $|\mathcal{R}_{n}\cap \ZZ|$, is unbounded as $n$ varies. Consequently, we can decide whether the loop admits a self-avoiding trace by inspecting the direction of this cone.

Statements \ref{two_dif_1} and \ref{two_dif_2} are the essential ones.

The case where the recession cone intersects $I$, as in \Cref{fig:two_inc}, corresponds to Statement~\ref{two_dif_1}. In this case, it points in a direction in which the absolute value of the state increases; thus, $\mathcal{R}$ admits a self-avoiding trace.

The case where the recession cone does not intersect $I\cup \Delta$, as in \Cref{fig:two_des}, corresponds to Statement~\ref{two_dif_2}. In this case, it points only in directions in which the absolute value of the state decreases; thus, $\mathcal{R}$ does not admit a self-avoiding trace.

The case in which the cone intersects $\Delta$ is slightly more delicate and requires analyzing the intersection $\mathcal{R}\cap I$, which is done in the remaining statements.

The complete proof is given in \Cref{proofs}.
\end{proof}

\begin{figure}[t]
\centering

\begin{subfigure}[t]{0.48\textwidth}
  \centering
  \vspace{0pt}%
  \begin{tikzpicture}[scale=0.9]
  \def\L{3.2}

  \draw[step=1, very thin] (-\L,-\L) grid (\L,\L);
  \draw[->] (-\L-0.2,0) -- (\L+0.2,0) node[right] {$x_1$};
  \draw[->] (0,-\L-0.2) -- (0,\L+0.2) node[above] {$x_2$};

  \draw[thin, dashed] (-\L,-\L) -- (\L,\L);

  \fill[blue!30, opacity=0.5] (0,0) -- (0,\L) -- (\L,\L) -- cycle;

  \fill[red!30, opacity=0.5] (-\L,-\L) -- (0,-\L) -- (0,0) -- cycle;

  \draw[very thick, green!60!black] (0,0) -- (\L,\L);
  \draw[very thick, green!60!black] (0,0) -- (-\L,-\L);

  \fill (0,0) circle (2pt);
  \node[below left] at (0,0) {\small $0$};

  \begin{scope}[shift={(1.5,-2.95)}]
    \draw[rounded corners, fill=white, opacity=0.9] (0,0) rectangle (1.5,1.55);
    \fill[blue!30, opacity=0.5] (0.15,1.20) rectangle (0.45,1.40);
    \node[anchor=west] at (0.55,1.30) {$I_+$};
    \fill[red!30, opacity=0.5] (0.15,0.85) rectangle (0.45,1.05);
    \node[anchor=west] at (0.55,0.95) {$I_-$};
    \draw[very thick, green!60!black] (0.15,0.55) -- (0.45,0.55);
    \node[anchor=west] at (0.55,0.55) {$\Delta$};
  \end{scope}

\coordinate (A) at (0,-1);
\coordinate (B) at (1,-2);

\pgfmathsetmacro{\tA}{min(\L, (\L+1)/4)} 
\coordinate (Aend) at ({0 + 1*\tA},{-1 + 4*\tA}); 

\pgfmathsetmacro{\tB}{min((\L-1)/2, (\L+2)/2)} 
\coordinate (Bend) at ({1 + 2*\tB},{-2 + 2*\tB}); 

\fill[gray, opacity=0.25] (A) -- (B) -- (Bend) -- (Aend) -- cycle;

\draw[thick] (A) -- (B);
\draw[thick,->] (A) -- (Aend);
\draw[thick,->] (B) -- (Bend);

\fill (A) circle (2pt);
\fill (B) circle (2pt);

\end{tikzpicture}

  \vspace{0.35em}
  \caption{SLC that points towards $I_+$.}
  \label{fig:two_inc}
\end{subfigure}%
\hfill%
\begin{subfigure}[t]{0.48\textwidth}
  \centering
  \vspace{0pt}%
  \begin{tikzpicture}[scale=0.9]
  \def\L{3.2}

  \draw[step=1, very thin] (-\L,-\L) grid (\L,\L);
  \draw[->] (-\L-0.2,0) -- (\L+0.2,0) node[right] {$x_1$};
  \draw[->] (0,-\L-0.2) -- (0,\L+0.2) node[above] {$x_2$};

  \draw[thin, dashed] (-\L,-\L) -- (\L,\L);

  \fill[blue!30, opacity=0.5] (0,0) -- (0,\L) -- (\L,\L) -- cycle;

  \fill[red!30, opacity=0.5] (-\L,-\L) -- (0,-\L) -- (0,0) -- cycle;

  \draw[very thick, green!60!black] (0,0) -- (\L,\L);
  \draw[very thick, green!60!black] (0,0) -- (-\L,-\L);

  \fill (0,0) circle (2pt);
  \node[below left] at (0,0) {\small $0$};

  \begin{scope}[shift={(1.5,-2.95)}]
    \draw[rounded corners, fill=white, opacity=0.9] (0,0) rectangle (1.5,1.55);
    \fill[blue!30, opacity=0.5] (0.15,1.20) rectangle (0.45,1.40);
    \node[anchor=west] at (0.55,1.30) {$I_+$};
    \fill[red!30, opacity=0.5] (0.15,0.85) rectangle (0.45,1.05);
    \node[anchor=west] at (0.55,0.95) {$I_-$};
    \draw[very thick, green!60!black] (0.15,0.55) -- (0.45,0.55);
    \node[anchor=west] at (0.55,0.55) {$\Delta$};
  \end{scope}

\coordinate (A) at (-2, 1);
\coordinate (B) at (-1,-1);

\pgfmathsetmacro{\tA}{min((\L+2)/5, (\L-1)/1)} 
\coordinate (Aend) at ({-2 + 5*\tA},{ 1 + 1*\tA}); 

\pgfmathsetmacro{\tB}{(\L+1)/4} 
\coordinate (Bend) at ({-1 + 4*\tB},{-1});

\fill[gray, opacity=0.25] (A) -- (B) -- (Bend) -- (Aend) -- cycle;

\draw[thick] (A) -- (B);
\draw[thick,->] (A) -- (Aend);
\draw[thick,->] (B) -- (Bend);

\fill (A) circle (2pt);
\fill (B) circle (2pt);

\end{tikzpicture}

  \vspace{0.35em}
  \caption{SLC that does not point towards $I\cup \Delta$.}
  \label{fig:two_des}
\end{subfigure}

\caption{}
\label{fig:two_diff}
\end{figure}

The next lemma, Lemma~\ref{one}, treats the case in which $\rec(\mathcal{R})$ is generated by a single vector. In this situation $\mathcal{R}$ is infinite but “thin”: each vertical slice $\mathcal{R}_{n}$ is bounded. We quantify this thinness via the $\height$.

\begin{definition}
   For $T \subseteq \RR^{2}$ and $p \in \NN$, define the \emph{$p$-height} of $T$ by\[\height_p(T) := \sup \{ |T_{z}\cap (\frac{1}{p}\ZZ)| \mid z\in \ZZ\}.\] 
\end{definition}

\begin{lemma} \label{one}
    Let $\mathcal{R}$ be an SLC. Suppose $\rec(\mathcal{R}) = \nonneg(v)$ and $v = (p,q)$ for $p\in \ZZ$, $q\in \ZZ$ and GCD$(p,q) =1$. Then:
    \begin{enumerate}
        \item \label{one_1} If $ |p| < |q|$, $\sign(p) = \sign(q)$ and $ \height_{|p|}( \mathcal{R}) \geq |p|$, then the loop has a self-avoiding trace.

        \item \label{one_2} If $\sign(p) \neq \sign(q)$, then the loop does not have a self-avoiding trace.

        \item \label{one_3} If $|p| > 1$ and $1 < \height_{|p|}( \mathcal{R}) < |p| $, then the loop does not have a self-avoiding trace if Conjecture \ref{simpleconj} holds.

        \item \label{one_4} If $|p| > 1$ and $\height_{|p|}( \mathcal{R}) \leq 1 $, then the loop does not have a self-avoiding trace.

        \item \label{one_5} If $\height_{|p|}( \mathcal{R}) = 0 $, then the loop does not have a self-avoiding trace.

        \item \label{one_6} If $|p| > |q|$, then the loop does not have a self-avoiding trace.

        \item \label{one_7} If $p = q = 1$  and $\mathcal{R}_\ZZ  \cap I_+ \neq \emptyset$, then the loop has a self-avoiding trace.

        \item \label{one_8} If $p = q = 1$  and $\mathcal{R}_\ZZ  \cap I_+ = \emptyset$, then the loop does not have a self-avoiding trace.

        \item \label{one_9} If $p = q = -1$  and $\mathcal{R}_\ZZ  \cap I_- \neq \emptyset$, then the loop has a self-avoiding trace.

        \item \label{one_10} If $p = q = -1$  and $\mathcal{R}_\ZZ  \cap I_- = \emptyset$, then the loop does not have a self-avoiding trace.        
    \end{enumerate}
\end{lemma}

\begin{proof}[Sketch of proof]
In this case, the cone is again pointed, but it is generated by a single vector. Since the cone is not “large” as in Lemma~\ref{two_diff}, it requires a more delicate analysis.

The cases where the recession cone does not intersect $I\cup \Delta$ (Statements~\ref{one_2} and~\ref{one_6}) are similar to Statement~\ref{two_dif_2} of Lemma~\ref{two_diff} and are handled in the same way.

The cases where the recession cone intersects $I$ are the essential ones and correspond to Statements~\ref{one_1}, \ref{one_3}, and~\ref{one_4}, which are the most important in this Lemma.

The case where the height of $\mathcal{R}$ is too small, as in \Cref{fig:Rthin}, corresponds to Statement~\ref{one_4}. In this regime, traces correspond to a weak Collatz sequence that visits only one residue class. Thus, the non-existence of a self-avoiding trace follows from Proposition~\ref{collatz_lemma}.

The case where the height of $\mathcal{R}$ is large enough, as in \Cref{fig:Rthick}, corresponds to Statement~\ref{one_1}. This case is analogous to Statement~\ref{two_dif_1} of Lemma~\ref{two_diff}. Here we show the existence of a self-avoiding trace.

The case where the height of $\mathcal{R}$ is at least $1$ but less than $|p|$, as in \Cref{fig:Rplus}, corresponds to Statement~\ref{one_3}. In this case, every trace is a weak Collatz sequence that visits fewer than $|p|$ residue classes. If the Reachability Conjecture~\ref{simpleconj} holds, then the loop does not admit a self-avoiding trace.

The complete proof is presented in \Cref{proofs}.
\end{proof}

\begin{figure}[t]
\centering

\begin{subfigure}[t]{0.48\textwidth}
  \centering
  \vspace{0pt}%
  \begin{tikzpicture}[scale=0.85]
    \def\xmin{0}
    \def\xstart{3}
    \def\xmax{7}
    \def\ymin{0}
    \def\ymax{10}

    \path[use as bounding box] (\xmin-0.9,\ymin-1.0) rectangle (\xmax+1.0,\ymax+1.0);

    \draw[step=1, very thin] (\xmin,\ymin) grid (\xmax,\ymax);

    \draw[->] (\xmin-0.3,0) -- (\xmax+0.5,0) node[right] {$x$};
    \draw[->] (0,\ymin-0.3) -- (0,\ymax+0.5) node[above] {$x'$};

    \foreach \t in {\xmin,...,\xmax} {
      \draw (\t,0.10) -- (\t,-0.10) node[below] {\small $\t$};
    }

    \foreach \t in {1,...,\ymax} {
      \draw (0.10,\t) -- (-0.10,\t) node[left] {\small $\t$};
    }

    \draw[thick,->] (\xstart,{(4*\xstart-1)/3}) -- (\xmax,{(4*\xmax-1)/3});
  \end{tikzpicture}

  \vspace{0.35em}
  \caption{Thin SLC.}
  \label{fig:Rthin}
\end{subfigure}%
\hfill%
\begin{subfigure}[t]{0.48\textwidth}
  \centering
  \vspace{0pt}%
  \begin{tikzpicture}[scale=0.85 ]
    \def\xmin{0}
    \def\xstart{3}
    \def\xmax{7}
    \def\ymin{0}
    \def\ymax{10}

    \path[use as bounding box] (\xmin-0.9,\ymin-1.0) rectangle (\xmax+1.0,\ymax+1.0);

    \draw[step=1, very thin] (\xmin,\ymin) grid (\xmax,\ymax);

    \draw[->] (\xmin-0.3,0) -- (\xmax+0.5,0) node[right] {$x$};
    \draw[->] (0,\ymin-0.3) -- (0,\ymax+0.5) node[above] {$x'$};

    \foreach \t in {\xmin,...,\xmax} {
      \draw (\t,0.10) -- (\t,-0.10) node[below] {\small $\t$};
    }

    \foreach \t in {1,...,\ymax} {
      \draw (0.10,\t) -- (-0.10,\t) node[left] {\small $\t$};
    }

    \draw[thick,->] (\xstart,{(4*\xstart-2)/3}) -- (\xmax,{(4*\xmax-2)/3});
    \draw[thick,->] (\xstart,{(4*\xstart)/3}) -- (\xmax,{(4*\xmax)/3});

    \fill[red, opacity=0.25]
      (\xstart,{(4*\xstart-2)/3}) --
      (\xmax,{(4*\xmax-2)/3}) --
      (\xmax,{(4*\xmax)/3}) --
      (\xstart,{(4*\xstart)/3}) -- cycle;
  \end{tikzpicture}

  \vspace{0.35em}
  \caption{Thick SLC.}
  \label{fig:Rthick}
\end{subfigure}

\caption{}
\label{fig:twofigs2}
\end{figure}

\begin{lemma} \label{two_col}
    Let $\mathcal{R}$ be an SLC. Suppose $\rec(\mathcal{R}) = \nonneg( v_1,v_2)$ and $v_1 = - v_2 = (p,q)$ for $p\in \ZZ_{\ge 0}$, $q\in \ZZ$ and GCD$(p,q) =1$. Then:
    \begin{enumerate}
        \item \label{two_col_1} If $ 0 < p < |q|$ and $\height_{p}( \mathcal{R}) \geq p$, then the loop has a self-avoiding trace.

        \item \label{two_col_2} If $p > 1$ and $1 < \height_{p}(  \mathcal{R}) < p$, then the loop does not have a self-avoiding trace if Conjecture \ref{simpleconj} holds.

        \item \label{two_col_3} If $p > 1$ and $\height_{p}( \mathcal{R}) \leq 1 $, then the loop does not have a self-avoiding trace.

        \item \label{two_col_4} If $\height_{|p|}(  \mathcal{R}) = 0 $, then the loop does not have a self-avoiding trace.

        \item \label{two_col_5} If $p > |q|$, then the loop does not have a self-avoiding trace.

        \item \label{two_col_6} If $p = q$  and $\mathcal{R}_\ZZ  \cap I \neq \emptyset$, then the loop has a self-avoiding trace.

        \item \label{two_col_7} If $p=q$ and $\mathcal{R}_\ZZ \cap I = \emptyset$, then the loop does not have a self-avoiding trace.

        \item \label{two_col_8} If $-p = q$  and $\height_1(\mathcal{R}) \geq 2$, then the loop has a self-avoiding trace.

        \item \label{two_col_9} If $-p = q$  and $\height_1(\mathcal{R}) \leq 1$, then the loop does not have a self-avoiding trace.  

         \item \label{two_col_10} If $p = 0$, then the loop does not have a self-avoiding trace. 
    \end{enumerate}
\end{lemma}

\begin{proof}[Sketch of proof]
This lemma treats the case where $\rec(\mathcal{R})$ is a line (i.e., generated by vectors $v_1$ and $-v_1$). In this setting, $\mathcal{R}$ is unbounded but “thin”: each vertical slice $\mathcal{R}_{n}$ is bounded. Accordingly, we analyze the height $\height_p(\mathcal{R})$ to determine whether a self-avoiding trace exists.

The intuition is the same as in Lemma~\ref{one}, but the argument requires some additional technicalities. The complete proof is presented in \Cref{proofs}.
\end{proof}

\begin{lemma} \label{simple_case}
Let $\mathcal{R}$ be an SLC. The following statements hold.
\begin{enumerate}
    \item If $\rec(\mathcal{R}) = \nonneg ( v_1,v_2,v_3)$ and it cannot be generated by two vectors, then the loop has a self-avoiding trace.
    \item If $\rec(\mathcal{R}) = \{ 0 \} $, then the loop does not have a self-avoiding trace.
\end{enumerate}
    
\end{lemma}

\begin{proof}[Sketch of proof]
    This Lemma analyzes the remaining two cases. The complete proof is presented in \Cref{proofs}.
\end{proof}

Finally, we are ready to prove the main theorem.

\begin{proof}[Proof of \Cref{main}]
In fixed dimension, integer feasibility for rational polyhedra can be decided in polynomial time~\cite{Lenstra1983IntegerPW}.

By \Cref{onevarcycle}, cycle detection for a one-variable SLC reduces to checking for cycles of length at most two. Each such check can be expressed as an ILP instance of dimension four, as in Proposition~\ref{bounded_ILP}, and can therefore be performed in polynomial time.

Moreover, a Minkowski--Weyl decomposition of $\mathcal{R}$ can be computed in polynomial time~\cite{10.5555/1370949}. By Lemmas~\ref{two_diff}, \ref{one}, \ref{two_col}, and \ref{simple_case}, the remaining analysis consists of a constant number of two-dimensional ILP feasibility checks, obtained by intersecting $\rec(\mathcal{R})$ with the cones $I_+, I_-, \Delta_+$, and $\Delta_-$, and, in certain cases, by computing the height.

It remains to argue that the height checks can also be carried out in polynomial time.
If $\rec(\mathcal{R})=\nonneg(v,-v)$, then all vertical slices $\mathcal{R}_z$ are translates of one another; hence, the height is independent of $z$ and can be computed from any single slice.

If $\mathcal{R}=\conv(w_{1},\dots,w_{\ell})+\nonneg(v)$ with $v=(p,q)$, $\sign(p)=\sign(q)$, and $\gcd(p,q)=1$, then the height of the slice $\mathcal{R}_z$ varies monotonically with $z$: it increases with $z$ when $p,q>0$ and decreases with $z$ when $p,q<0$.

Since $\conv(w_{1},\dots,w_{\ell})$ is bounded, there exists $M\in\NN$ such that the absolute values of the coordinates of all points in $\conv(w_{1},\dots,w_{\ell})$ are bounded by $M$.

Assume $p,q>0$. Then for any $z_2\ge z_1$ we have
$\mathcal{R}_{z_1} + (z_2-z_1)\frac{q}{p}\ \subseteq\ \mathcal{R}_{z_2}.$
Moreover, if $z_1\ge M$, then $\mathcal{R}_{z_1} + (z_2-z_1)\frac{q}{p}\ =\ \mathcal{R}_{z_2}.$
Hence, it suffices to check the height at $z=M$.

If $p,q<0$, a symmetric argument shows that it suffices to check the height at $z=-M$.
\end{proof}

\begin{remark}
The algorithm in \Cref{main} can be used as a semi-algorithm without assuming that the Reachability Conjecture holds: one can repeat the case distinction from Lemmas~\ref{two_diff}, \ref{one}, \ref{two_col}, and \ref{simple_case}. However, the procedure fails to return an answer in Case~3 of Lemma~\ref{one} and Case~2 of Lemma~\ref{two_col}.
\end{remark}

\section{Proofs}\label{proofs}

\subsection{Cycles}


First, we state an auxiliary lemma.

\begin{lemma} \label{auxlemma}
    Let $a,\alpha,b$ be integers. If $sign(a-\alpha) = sign(b-\alpha) \neq 0 $, then $$(\frac{ba-\alpha^2}{a+b-2\alpha}, \frac{ba-\alpha^2}{a+b-2\alpha}) \in \conv\Bigl((a,\alpha), (\alpha,b)\Bigr).$$ 
\end{lemma}

\begin{proof}
    Define $x = \frac{b-\alpha}{a+b-2\alpha} $. Note that $0<x<1$ because $sign(a-\alpha) = sign(b-\alpha)\neq 0$.

    Hence we obtain
    $$x(a,\alpha) + (1-x)(\alpha,b) \in \conv\Bigl((a,\alpha), (\alpha,b)\Bigr).$$
    Let us rewrite the expression.

    \[x(a,\alpha) + (1-x)(\alpha,b) = (\frac{(b-\alpha)a + (a-\alpha)\alpha}{b+a-2\alpha},\frac{(b-\alpha)\alpha + (a-\alpha)b}{b+a-2\alpha}) = \] \[ = (\frac{ba-\alpha^2}{b+a-2\alpha},\frac{ba-\alpha^2}{b+a-2\alpha}).\]
    
\end{proof}

\begin{proof}[Proof of \Cref{onevarcycle}]

    Suppose $\mathcal{R}$ has a cycle $C = \{s_1,\dots,s_m\}$, where all $s_i$ are different states (otherwise we can take a smaller cycle) and $m>2$. Then for every $1\leq i < m$: $(s_i,s_{i+1}) \in \mathcal{R}$ and $(s_m,s_1)\in \mathcal{R}$, so
    $$(\frac{\sum_{i=1}^m s_i}{m},\frac{\sum_{i=1}^m s_i}{m})\in \mathcal{R}.$$Denote $z := \frac{\sum_{i=1}^m s_i}{m}$. Consider the translated relation
$\mathcal{R}' := \{(x-\lfloor z\rfloor,y-\lfloor z\rfloor)\mid (x,y)\in \mathcal{R}\}$ and translated states $s_i' := s_i-\lfloor z\rfloor$.
Then $\mathcal{R}'$ has a cycle $(s_1',\dots,s_m')$, and the corresponding average satisfies $0\le z-\lfloor z\rfloor<1$.
Thus, without loss of generality, we may assume $0\le z<1$. If $z = 0$, then there is a cycle of length $1$. Suppose this is not the case.

    Let $\alpha$ denote the state of $C$ with the largest absolute value, and let $a$ and $b$ be the neighboring states of $\alpha$ in $C$. Hence, $(a,\alpha)$ and $(\alpha,b)$ are in $\mathcal{R}$. Since $|a|,|b|\leq|\alpha|$ and all $s_i$ are different, we have $sign(a-\alpha) = sign(b-\alpha)\neq 0$. Now we can apply Lemma \ref{auxlemma} and obtain
    $$(\frac{ba-\alpha^2}{a+b-2\alpha}, \frac{ba-\alpha^2}{a+b-2\alpha}) \in \conv\Bigl((a,\alpha), (\alpha,b)\Bigr)  \subseteq  \mathcal{R}.$$
    Denote $\beta := \frac{ba-\alpha^2}{a+b-2\alpha}$ and consider the following cases.

    \smallskip
\noindent\textbf{Case 1: $ \beta \le 0$ or $1\le \beta$ .} In this case, since $0\leq z < 1$, either $(0,0)$ or $(1,1)$ can be expressed as a convex combination of $(z,z)$ and $(\beta,\beta)$, yielding a cycle of length one.

    \smallskip
\noindent\textbf{Case 2: $\alpha\le 0 $.} 
Note that $\alpha \le a,b \le -\alpha$ and $a,b$ and $\alpha$ are distinct. Thus, $a+b-2\alpha > 0$, and $ba-\alpha^2 < 0$, which implies $\beta < 0$. Hence, the statement follows from Case 1.

\smallskip
\noindent\textbf{Case 3: $0<\alpha$ and $0<\beta < 1$.}
    Since  $2\alpha > a+b$, rewriting $\beta<1$, we obtain
    $$a+b-2\alpha < ab-\alpha^2, $$
    $$\alpha^2 - 2\alpha < ab-a-b,$$
    $$(\alpha-1)^2 < (a-1)(b-1).$$
    If, for some $i$, $s_i = -\alpha$, then we can take $-\alpha$ instead of $\alpha$ and the statement follows from Case 2.
    Thus, we can assume that for every $i$ we have $-\alpha <s_i$. Therefore, 
    $$|a|,|b| < \alpha,$$ 
    $$-\alpha < a,b < \alpha, $$ 
    $$-\alpha - 1 < a-1,b-1 < \alpha - 1.$$
    Hence inequality $(\alpha-1)^2 < (a-1)(b-1)$ can be satisfied only if $a=-\alpha+1$ and $b=-\alpha+2$ or $a=-\alpha+2$ and $b=-\alpha+1$. Without loss of generality, assume the first case.

    Then $a$ is the smallest value in $C$. Let $c$ denote the predecessor of $a$ in $C$.

     All states in $C$ are different, and the length of $C$ is at least $3$, hence $-\alpha +2 \leq c<\alpha$. Since $sign(\alpha - a) = sign(c - a) = 1$, then we can apply Lemma \ref{auxlemma} to get

    $$(\frac{\alpha c - a^2}{c+\alpha - 2a},\frac{\alpha c - a^2}{c+\alpha - 2a}) \in \conv\Bigl((a,\alpha), (c,a)\Bigr) \subseteq \mathcal{R},$$

    $$(\frac{\alpha c - (\alpha-1)^2}{c+3\alpha - 2},\frac{\alpha c - (\alpha-1)^2}{c+3\alpha - 2}) \in \mathcal{R}.$$
    If $c<\alpha - 1$, then $\frac{\alpha c - (\alpha-1)^2}{c+3\alpha - 2} \leq 0$, hence $(0,0) \in \conv\Bigl((\frac{\alpha c - (\alpha-1)^2}{c+3\alpha - 2},\frac{\alpha c - (\alpha-1)^2}{c+3\alpha - 2}),(z,z)\Bigr) \subseteq \mathcal{R}$.

    It remains to prove the statement for $c= \alpha - 1$.
    Since $(\alpha-1, 1-\alpha)$ and $(\alpha, 2-\alpha)$ are in $\mathcal{R}$ we get that
    $$(\alpha - \frac{1}{2},\frac{3}{2} - \alpha) \in \mathcal{R}$$
    It is easy to see that points $(1,0)$ and $(0,1)$ lie on the interval
    $$[ (\alpha - \frac{1}{2},\frac{3}{2} - \alpha); (1-\alpha,\alpha)] \subseteq \mathcal{R},$$
    because $(\alpha - \frac{1}{2},\frac{3}{2} - \alpha) = (\frac{1}{2},\frac{1}{2}) + (\alpha-1, 1-\alpha)$ and $(1-\alpha,\alpha) = (\frac{1}{2},\frac{1}{2}) + (\frac{1}{2}-\alpha,\alpha - \frac{1}{2})$.
    Hence, $\mathcal{R}$ has a cycle of length 2.
    
\end{proof}

\subsection{Generalized Collatz Sequences}


\begin{proof} [Proof of \Cref{collatz_lemma}]
    Assume the contrary, there exist $T$, $n$ and $i$ such that $T^k(n)\equiv i\text{ (mod } d)$ for every $k\geq 0 $. Hence, \[T^{k+1}(n) = \frac{m_i T^{k}(n) - r_i}{d}.\] We prove by induction on $k$, that for every $k\geq 1$
    \[T^k(n) = \frac{nm_i^k - r_i\sum_{j=0}^{k-1}m_i^jd^{k-1-j}}{d^k}.\]
    For $k=1$ we get $T(n)=\frac{nm_i-r_i}{d}$. Assuming the formula for $k$ let us prove it for $k+1$,
    \[T^{k+1}(n) = \frac{m_iT^k(n)-r_i}{d} = \frac{m_i(nm_i^k - r_i\sum_{j=0}^{k-1}m_i^jd^{k-1-j}) - r_i d^k}{d^{k+1}} = \] \[= \frac{nm_i^{k+1} - r_i \sum_{j=0}^{k}m_i^jd^{k-j}
}{d^{k+1}}.
    \]
    Hence, for every $k\ge 1$,
    \[nm_i^k - r_i\sum_{j=0}^{k-1}m_i^jd^{k-1-j} \equiv 0 \text{ (mod }d^k).\]
    Multiplying it by $m_i-d$ yields
    \[(m_i-d)(nm_i^k - r_i\sum_{j=0}^{k-1}m_i^jd^{k-1-j}) = nm_i^{k+1} - ndm_i^k - r_i(m_i^k -d^k) \equiv 0 \text{ (mod }d^k). \]
    Since $d^k \equiv 0 \text{ (mod }d^k)$ and $\gcd(m_i,d) = 1$, we obtain for every $k$ that
    \[nm_i - nd -r_i \equiv 0 \text{ (mod }d^k).\]
    As this holds for all $k$ and $d>1$, it follows that $nm_i - nd -r_i=0$. Thus $n = \frac{nm_i-r_i}{d}$. Therefore $T^k(n) =n$ for all $k$, contradicting unboundedness.
\end{proof}



\begin{proof} [Proof of \Cref{reduction}]
    Recall that $T$ has the form
    $$T(x) = 
\frac{m x- (a+i)}{d}, \ \ \ \ \text{if }  mx\equiv a+i \ (\text{mod } d)$$
Define 
\[\mathcal{R}^+ := \{(x,x')\in \RR^2\mid mx-a-d+1 \leq dx' \leq mx-a-1 \text{ and } 0<x < x' \}\]
Note that,
$$\mathcal{R}^+_\ZZ = \{(n,T(n)) \in \ZZ^2\mid 0<n<T(n) \text{ and }mn \not\equiv a \ (\text{mod }d )\}.$$

Therefore every infinite trace of $\mathcal{R}^+$ is self-avoiding and corresponds to the unbounded sequence $\{mT^k(n) \}_{k\in \NN}$ with $mT^k(n)\not\equiv a \ (\text{mod }d )$, and $T^k(n) > 0$, for every $k\ge 0$. 

For the other direction, if there exists an infinite unbounded sequence $\{T^k(n) \}_{k\in \NN}$ of positive values, then, since $m>d$, from some step it takes values large enough to have $T^k(n) < T^{k+1}(n)$. It implies the existence of a self-avoiding trace of $\mathcal{R}^+$.

The proof for 
\[\mathcal{R}^- := \{(x,x')\in \RR^2\mid mx-a-d+1 \leq dx' \leq mx-a-1 \text{ and } 0> x > x' \}\]
is symmetric.
\end{proof}

\subsection{Self-avoiding Traces}


Let us formulate an auxiliary notation and an auxiliary lemma. 
 \[C_+ = \{ (x_1,x_2)\in \RR^2 \mid x_1 \ge 0 \ge x_2\}\ \ \ \ \ \ \ \ \ \ C_- = \{ (x_1,x_2)\in \RR^2 \mid x_2\ge 0 \ge x_1\}\]

\begin{lemma}\label{auxprop}
    Let $v_1,v_2, w,u \in \RR^2$ be nonzero vectors with $v_1\nparallel v_2$, $w\in C_+$, $u\in C_-$  and $w,u\in \nonneg(v_1,v_2)$. Then $\nonneg(v_1,v_2)\cap \Delta \neq \emptyset$ and $\nonneg(v_1,v_2)\cap I \neq \emptyset$.
\end{lemma}

 \begin{proof}
    If $w \parallel u$, then $nonneg(w,u)$ is a line contained in the closed cone $nonneg(v_1,v_2)$, contradicting Proposition \ref{classification}, since $v_1\nparallel v_2$. Thus $w\nparallel u$. 

    Let us prove that there exists $\alpha\in\RR_+$, such that $\alpha w + (1-\alpha) u \in \cl(\Delta)$. Denote $w= (w_1,w_2)$ and $u = (u_1, u_2)$. Define $\alpha := \frac{u_2-u_1}{w_1 - w_2 + u_2 - u_1}$. Since $w\in C_+$ and $u\in C_-$, we obtain $w_1-w_2 \geq 0$, $u_2-u_1 \geq 0$ and since $w\nparallel u$ they cannot be both equal to $0$. Hence denominator of $\alpha$ is nonzero and $\alpha\geq 0$.  A direct computation gives 
    \[\alpha w + (1-\alpha) u     = w \frac{u_2-u_1}{w_1 - w_2 + u_2 - u_1} + u \frac{w_1-w_2}{w_1 - w_2 + u_2 - u_1} = \] \[(\frac{w_1(u_2-u_1)+ u_1 (w_1-w_2)}{w_1 - w_2 + u_2 - u_1}, \frac{w_2(u_2-u_1)+ u_2 (w_1-w_2)}{w_1 - w_2 + u_2 - u_1}) = \]
    \[ = (\frac{w_1u_2 - u_1w_2}{w_1 - w_2 + u_2 - u_1}, \frac{w_1u_2 - u_1w_2}{w_1 - w_2 + u_2 - u_1})\in \cl(\Delta).\]
    
    Since $w\nparallel u$ we have $\alpha w + (1-\alpha) u \neq (0,0)$, thus $\alpha w + (1-\alpha) u\in \Delta$. Recall that $u,w \in nonneg(v_1,v_2)$, hence $\alpha w + (1-\alpha) u \in \Delta$.

    Since $\alpha w + (1-\alpha) u \nparallel v_1$, and $\alpha w + (1-\alpha) u \nparallel v_2$, we get that $\alpha w + (1-\alpha) u $ is in the interior of $\nonneg(v_1,v_2)$, thus, some open neighborhood of $\alpha w + (1-\alpha) u $ is contained in $\nonneg(v_1,v_2)$, but any such neighborhood intersects $I$, since $\Delta \subseteq \cl(I)$. Hence, the statement follows.
\end{proof}


    
    

    

     

\begin{proof} [Proof of \Cref{two_diff}]
    \begin{enumerate}
        \item Without loss of generality assume $\nonneg(v_1,v_2) \cap I_+ \neq \emptyset$. The case when $\nonneg(v_1,v_2) \cap I_- \neq \emptyset$ follows from the considered one applied to $-\mathcal{R}$. 
        
        The set $\nonneg( v_1,v_2 ) \cap I_+$ is a (possibly nonclosed) convex cone as the intersection of two convex cones. Let us prove that it contains at least two noncollinear vectors. Indeed, for the contradiction, suppose there exists a vector $v$, such that every vector in $\nonneg(v_1,v_2) \cap I_+$ is collinear to  $v$. Thus, $ \nonneg(v_1,v_2) \cap I_+ \subseteq \nonneg(v)$. 
        

        Informally, $\nonneg(v_1,v_2)$ contains "half of a neighborhood of $v$".
        Formally, for angles $-\pi < \epsilon_1 < \epsilon_2 < \pi$, let $rot_v(\epsilon_1,\epsilon_2)$ be the set of vectors $v_{\mathrm{rot}}\in\RR^2$ such that there exists an angle $\theta$ with $\epsilon_1<\theta<\epsilon_2$ and
        $v_{\mathrm{rot}}$ is the image of $v$ under rotation by $\theta$ around the origin.
        Equivalently, $rot_v(\epsilon_1,\epsilon_2)=\{R_\theta v \mid \epsilon_1<\theta<\epsilon_2\}$, where $R_\theta$ is the rotation matrix.

Since $v\in\nonneg(v_1,v_2)$ and $v_1\nparallel v_2$, there exists $\epsilon\in(0,\pi)$ such that either
$rot_v(0,\epsilon)\subseteq \nonneg(v_1,v_2)$ or $rot_v(-\epsilon,0)\subseteq \nonneg(v_1,v_2)$.

Moreover, since $I_+$ is open and $v\in I_+$, we can choose $\epsilon>0$ small enough such that
$rot_v(-\epsilon,\epsilon)\subseteq I_+$.
Consequently, one of the sets $rot_v(0,\epsilon)$ or $rot_v(-\epsilon,0)$ is contained in
$\nonneg(v_1,v_2)\cap I_+$, contradicting the assumption that
$\nonneg(v_1,v_2)\cap I_+ \subseteq \nonneg(v)$.
        

        
        Hence, there are two non-collinear vectors $w,u\in \nonneg( v_1,v_2 ) \cap I_+$. Then \[|( \nonneg( v_1,v_2 ) \cap I_+)_{n}\cap\ZZ | \ge |\nonneg(w,u)_{n} \cap \ZZ|\ge n|\frac{w}{|w|} - \frac{u}{|u|}| \xrightarrow[n\rightarrow\infty]{} \infty.\] 
        Hence, for sufficiently large $n$, there always exists a transition from $n$ to a larger value, thus the loop has a self-avoiding trace. 

        \item Suppose there is a self-avoiding trace $\{ s_i\}_{i\in \NN}$.
        If there are infinitely many negative and infinitely many positive states, then infinitely many indices $i$ satisfy $s_i<0<s_{i+1}$ and infinitely many satisfy $s_i>0>s_{i+1}$. Thus, sets $\mathcal{R} \cap C_+$ and $\mathcal{R} \cap C_-$ are unbounded. In that case, their recession cones are non-empty, thus, by Proposition \ref{infcap}, there are non-zero vectors $w \in \rec(\mathcal{R}) \cap  C_+ $ and $u \in \rec(\mathcal{R}) \cap C_-$. Hence, by Lemma \ref{auxprop}, $\rec(\mathcal{R}) \cap \Delta \neq \emptyset$. Contradiction.

        Suppose instead there are only finitely many negative $s_i$.  The case when there are only finitely many positive $s_i$ follows from the considered one for $-\mathcal{R}$. 
        
        Then there are infinitely many $i$ such that $s_{i+1} > s_i > 0$. 
        By Theorem \ref{Minkowski-Weyl}, $\mathcal{R} = \conv(w_1,\dots, w_l) + \rec(\mathcal{R})$. Hence, for every $i$, there are $w^i\in \conv(w_1,\dots, w_l)$ and $v^i\in \rec(\mathcal{R})$, such that  $(s_i,s_{i+1}) = w^i + v^i$. Consider the normalized vectors $\frac{v^i}{\lVert v^i \rVert}$ along the infinite subsequence with $s_{i+1} > s_i > 0 $. Since it is an infinite sequence on a compact sphere, it must have a converging subsequence $\frac{v^{i_j}}{\lVert v^{i_j} \rVert}$. Since  for every $j$ holds $\frac{v^{i_j}}{\lVert v^{i_j} \rVert}\in \rec(\mathcal{R})$ this sequence must converge to $v\in \rec(\mathcal{R})$ and $\lVert v \rVert=1$. Let us prove that $v\in \cl(I_+)$. 
        
        Assume that $v\notin \cl(I_+)$. Since $\frac{v^{i_j}}{\lVert v^{i_j} \rVert }$ converges to $v$, for every $\epsilon$ there exists $N$, such that for every $j>N$, $\lVert\frac{v^{i_j}}{\lVert v^{i_j} \rVert} - v \rVert \leq \epsilon$. Let us take $\epsilon < \dist(v, \cl(I_+))$. Thus, $\nonneg(v'\in \RR^2 \mid \lVert v-v'\rVert \leq \epsilon) \cap \cl(I_+) = \emptyset$, hence by Proposition \ref{infcap} $\rec(\nonneg(v'\in \RR^2 \mid \lVert v-v'\rVert \leq \epsilon)) \cap \rec(\conv(-w_1,\dots,-w_l) + I_+ ) = \emptyset$, which implies that $\nonneg(v'\in \RR^2 \mid \lVert v-v'\rVert \leq \epsilon) \cap (\conv(-w_1,\dots,-w_l) + I_+ )$ is bounded by \Cref{Minkowski-Weyl}. For every $j>N$, $(s_{i_j}, s_{i_j +1}) - w^{i_j} = v^{i_j}  \in \nonneg(v'\in \RR^2 \mid \lVert v-v'\rVert \leq \epsilon)$, but $(s_{i_j},s_{i_j+1})\in I_+$, hence $v^{i_j}\in \conv(-w_1,\dots,-w_l) + I_+ $. Since the sequence $s_{i_j}$ is unbounded and the set $\conv(w_1,\dots,w_l)$ is bounded, the sequence $v^{i_j}$ is also unbounded, but it is contained in a bounded set. Thus $v\in \cl(I_+)$ by contradiction.

        Since $v\in \cl(I_+)\cap \rec(\mathcal{R})$ and $\rec(\mathcal{R}) \cap (I_+\cup \Delta_+) = \emptyset$ we obtain $v\in \{ (0,x_2)\in \RR^2 \mid x_2 > 0\}$, hence $v=(0,1)$. There exists $N$, such that for every $j>N$ holds $|\frac{v^{i_j}}{\lVert v^{i_j} \rVert} - (0,1)| < 0.1$, hence $\frac{v^{i_j}}{\lVert v^{i_j} \rVert} \in \{(x_1,x_2) \mid x_1\leq 0 \}$, because otherwise $\frac{v^{i_j}}{\lVert v^{i_j} \rVert} \in I_+$, which is impossible, since $\frac{v^{i_j}}{\lVert v^{i_j} \rVert} \in \rec(\mathcal{R})$. Since the set $\conv(w_1,\dots,w_l)$ is bounded, the supremum $\sup\{ x_1\in\RR \mid \exists x_2\in \RR: (x_1,x_2)\in \conv(w_1,\dots,w_l)\} $ is finite, denote it as $K$. Since $s_{i_j} = w^{i_j} + v^{i_j}$ from the proven above, for every $j>N$, $0< s_{i_j} < K$. That contradicts the distinctness of $s_i$ and the assumption that there are only finitely many negative $s_i$.

        \item If $\rec(\mathcal{R}) \cap \Delta_+ \neq \emptyset$, then $\Delta_+ \subseteq \rec(\mathcal{R}) $. If $(a,b) \in \mathcal{R}_{\ZZ} \cap I_+$, then $(a,b) + \Delta_+ \subseteq \mathcal{R}$, hence for any $n\geq a$ it is always possible to make a transition to $n-a+b$. Thus, the loop has a self-avoiding trace.
        
        \item Suppose there is a self-avoiding trace $\{ s_i \}_{i\in \NN}$. If there are infinitely many negative and infinitely many positive states, then infinitely many indices $i$ satisfy $s_i<0<s_{i+1}$ and infinitely many satisfy $s_i>0>s_{i+1}$. Thus, sets $\mathcal{R} \cap C_+$ and $\mathcal{R} \cap C_-$ are unbounded. In that case, their recession cones are non-empty, thus, by Proposition \ref{infcap}, there are non-zero vectors $w \in \rec(\mathcal{R}) \cap  C_+ $ and $u \in \rec(\mathcal{R}) \cap C_-$. Hence, by Lemma \ref{auxprop}, $\rec(\mathcal{R}) \cap I \neq \emptyset$. Contradiction.

        There cannot be only finitely many negative $s_i$, due to the argument identical to the argument in statement~\ref{two_dif_2} of this lemma. 
        
        If it contains only finitely many positive terms, then there is some index with $0>s_i>s_{i+1}$, contradicting $\mathcal{R}_{\ZZ}\cap I_-=\emptyset$.

        \item This statement follows from the statement \ref{two_dif_3} applied to $-\mathcal{R}$.

        \item This statement follows from the statement \ref{two_dif_4} applied to $-\mathcal{R}$.
    \end{enumerate}
\end{proof}

\begin{proof} [Proof of \Cref{one}]
    \begin{enumerate}
        \item Assume $0 < p < q$. The case when $q < p < 0$ follows from the considered one applied to $-\mathcal{R}$. 

        If $\height_{|p|}( \mathcal{R}) \geq p$, then there exists $z\in \ZZ$, such that \[|\mathcal{R}_{z}\cap (\frac{1}{p}\ZZ)  | \geq p.\]

        For any $z'\ge z$, using $(p,q)\in \rec(\mathcal{R})$ we can shift by $(z'-z)/p$ along $(p,q)$, which maps $x_1=z$ to $x_1=z'$ and adds $\frac{q}{p}(z'-z)$ to the second coordinate. Hence

        \[|\mathcal{R}_{z'}\cap (\frac{1}{p}\ZZ)  | \geq |\mathcal{R}_{z}\cap (\frac{1}{p}\ZZ) + \frac{q}{p}(z'-z) | = |\mathcal{R}_{z}\cap (\frac{1}{p}\ZZ)  | \geq p.\]

        Since $|\mathcal{R}_{z'}|$ is convex $|\mathcal{R}_{z'}\cap (\frac{1}{p}\ZZ)| \geq p$ implies that $\mathcal{R}_{z'}$ contains at least $p$ consecutive points from $\frac{1}{p}\ZZ$, hence at least one of them is an integer. Thus $|\mathcal{R}_{z'}\cap \ZZ  | \geq 1$. 

        Since $0<p < q$ there exists $N \in \NN$, such that for every $z' \geq N$, $\mathcal{R}_{z'} \subseteq I_+$. Hence, $\mathcal{R}_{z'}\cap \ZZ \subseteq [z'+1; \infty)$. Thus, for every $z' > \max(N,z)$, it is possible to make a transition to a larger value. Hence, the loop has a self-avoiding trace.

        \item Assume $p< 0 < q$. The case when $q  < 0 <p$ follows from the considered one applied to $-\mathcal{R}$. In that case,
        \[\rec(\mathcal{R}) \subseteq \RR_-\times\RR_+ \cup \{ (0,0)\}.\]By Proposition \ref{infcap} the set $\mathcal{R} \cap (\RR^2 \setminus (\RR_- \times \RR_+))$ is bounded. Therefore, any self-avoiding trace would eventually have to make transitions only from $\RR_{-}$ to $\RR_{+}$, which is impossible. Hence, no self-avoiding trace exists.

        \item Assume $0 < p , q$. The case when $p,q < 0$ follows from the considered one applied to $-\mathcal{R}$. The case when $\sign(p)\neq \sign(q)$ follows from the statement \ref{one_2} of this Lemma. 

        Denote $h = \height_{|p|}( \mathcal{R})$. There exists $z\in \ZZ$, such that $|\mathcal{R}_{z}\cap (\frac{1}{p}\ZZ)  |=h.$ Since $\mathcal{R}$ is convex the points of $\mathcal{R}_{z}\cap (\frac{1}{p}\ZZ)$ are of the form $\{ \frac{qz - a - i}{p} \}_{1 \leq i \leq h}$ for some $a\in \ZZ$.  

        Let us prove that for all $z'\in \ZZ$,
        \begin{equation}\label{incl} \mathcal{R}_{z'} \cap (\frac{1}{p}\ZZ) \subseteq \mathcal{R}_{z} \cap (\frac{1}{p}\ZZ) + \frac{q}{p}(z'-z).
        \end{equation}

        For $z' > z$ since $(p,q)\in \rec(\mathcal{R})$ and $z'-z>0$ the following holds 
        \[\mathcal{R}_{z} \cap (\frac{1}{p}\ZZ) + \frac{q}{p}(z'-z) \subseteq \mathcal{R}_{z'} \cap (\frac{1}{p}\ZZ),\]
        but since $|\mathcal{R}_{z'} \cap (\frac{1}{p}\ZZ)| \leq h$ we get (\ref{incl}). 
        
        Suppose that $z' < z$, then for any $r \in \mathcal{R}_{z'}\cap (\frac{1}{p}\ZZ)$, since $(p,q)\in \rec(\mathcal{R})$ and $z-z'>0$, we obtain $r+\frac{q}{p}(z-z')\in \mathcal{R}_{z}\cap (\frac{1}{p}\ZZ)$. Thus, we get (\ref{incl}).
        
        Hence $\mathcal{R}_\ZZ \subseteq \{ (x,y) \in  \ZZ^2 \mid y = \frac{qx-a-i}{p}, \ 1\leq i\leq h  \}.$
        
        Non-existence of a self-avoiding trace follows from Conjecture \ref{simpleconj}, since $h<p$. 

        \item The proof is analogous to the proof of the statement \ref{one_3}, but at the end, instead of Conjecture \ref{simpleconj}, Proposition \ref{collatz_lemma} needs to be applied.

        \item In that case $\mathcal{R}\cap \ZZ^2 = \emptyset$, so the claim is immediate.

        \item If $|p| > |q|$, then$$\rec(\mathcal{R})\subseteq \{ (x_1,x_2)\in \RR^2 \mid |x_1| > |x_2| \text{ or } x_1=x_2=0\}.$$Hence, the set $\mathcal{R} \cap \{ (x,y)\in \RR^2 \mid |x| \leq |y|\}$ is bounded. Assume there is a self-avoiding trace $\{ s_i\}_{i\in \NN}$. Then there are infinitely many $i$, such that $|s_i| < |s_{i+1}|$. Contradiction.
        
    \end{enumerate}

    If $p=q > 0$, then $\rec(\mathcal{R})= \Delta_+$. If $p=q < 0$, then $\rec(\mathcal{R}) = \Delta_-$. Hence, the proofs of statements \ref{one_6},\ref{one_7},\ref{one_8} and \ref{one_9} are analogous to the proofs of statement \ref{two_dif_3}, \ref{two_dif_4}, \ref{two_dif_5} and \ref{two_dif_6} of Lemma \ref{two_diff} .
\end{proof}

\begin{proof} [Proof of \Cref{two_col}]
    The proofs of \ref{two_col_1}, \ref{two_col_2}, \ref{two_col_3}, \ref{two_col_4}, \ref{two_col_5}, \ref{two_col_6}, and \ref{two_col_7} are analogous to the proofs of Lemma~\ref{one}, items \ref{one_1}, \ref{one_3}, \ref{one_4}, \ref{one_5}, \ref{one_6}, \ref{one_7}, and \ref{one_8}, respectively.

    Assume $-p=q$. Since $\gcd(p,q)= 1$ we have $p=1$, $q=-1$.  
    
    Denote $h=\height_1(\mathcal{R})$. Take $z\in \ZZ$, such that $|\mathcal{R}_{z}\cap \ZZ | = h $. For every $z'\in \ZZ$,

    \[\mathcal{R}_{z} \cap \ZZ + (z'-z) \subseteq \mathcal{R}_{z'} \cap \ZZ , \]
    and since $|\mathcal{R}_{z'} \cap \ZZ  | \leq h$, equality holds
    \[\mathcal{R}_{z} \cap \ZZ + (z'-z) = \mathcal{R}_{z'} \cap \ZZ . \]

    Denote $\mathcal{R}_{0}\cap \ZZ = [a;b]$, then $h = b-a + 1$. Thus $\mathcal{R}_\ZZ = \{ (x,y)\in \ZZ \mid a\leq y+x \leq b\}$. If $h = 1$, then for every trace $s_1,s_2,s_3$ we have $s_1 + s_2 = s_2 + s_3 = a$, hence $s_1=s_3$, the loop does not have a self-avoiding trace, thus statement \ref{two_col_9} is proved.
    
    If $h\geq 2$, then let us define a trace $\{ s_i\}_{i\in \NN}$ by \[
s_1 = 2  |a|, \qquad
s_{i+1} =
\begin{cases}
  a - s_i, & \text{if } i \text{ is odd}, \\[6pt]
  b - s_i, & \text{if } i \text{ is even}.
\end{cases}
\]

For every odd $i$ we get  $s_{i+2} = b-a + s_i > s_i$ and $s_{i+1} = a-s_i < a-s_1 < a-2|a| < 0$, hence this trace is self-avoiding. Thus, statement \ref{two_col_8} is proved.

The statement \ref{two_col_10}  follows from the fact, that if $p = 0$, then the set $\{ x\in \ZZ \mid \exists y\in \ZZ.  \ (x,y)\in \mathcal{R} \}$ is bounded, thus an infinite trace can visit only the finite number of different states.

\end{proof}

\begin{proof} [Proof of \Cref{simple_case}]
    \begin{enumerate}
        \item If $\rec(\mathcal{R}) = \nonneg ( v_1,v_2,v_3)$ and this cone cannot be generated by two vectors, then by Proposition \ref{classification} it is either a whole plane or a (closed) affine half-plane. In the former case, $\mathcal{R} = \RR^2$. For the latter case, note that any closed affine half-plane must intersect $I$, hence the argument of Lemma~\ref{two_diff}.\ref{two_dif_1}  applies and yields a self-avoiding trace.
        \item If $\rec(\mathcal{R}) = \{ 0 \}$, then $\mathcal{R}$ is bounded, hence any infinite trace can visit only a finite number of states and thus cannot be self-avoiding.
    \end{enumerate}
\end{proof}

\section{Conclusion}\label{conclusion}
In this paper, we have studied the termination problem for SLCs over $\mathbb{Z}$. We have observed that a loop has an infinite computation iff it admits either a cycle or a self-avoiding trace.

We first analyzed cyclic traces. For one-variable SLCs, we have proved that the existence of a cycle implies the existence of a cycle of length at most two. An open question is whether, for $n>1$, there exists a bound $K(n)$ such that any $n$-variable loop that has a cycle also has one of length at most $K(n)$.

Next, we have formulated the Reachability Conjecture for weak Collatz sequences -- a special case of a long-standing conjecture -- and proved it for modulus $d=2$. We have shown that a decision procedure for one-variable SLC termination would verify the Reachability Conjecture for specific mappings; for most mappings (in particular, all cases with $d>2$), the conjecture remains open.

Finally, assuming the Reachability Conjecture, we have established that termination of one-variable SLCs is decidable in polynomial time.



\bibliography{lipics-v2021-sample-article}

@article{Matthews1984,
author = {Matthews K., Watts A.},
journal = {Acta Arithmetica},
language = {eng},
number = {2},
pages = {167-175},
title = {A generalization of Hasse's generalization of the Syracuse algorithm},
url = {http://eudml.org/doc/205897},
volume = {43},
year = {1984},
}

@article{10.1145/2400676.2400679,
author = {Ben-Amram, Amir M. and Genaim, Samir and Masud, Abu Naser},
title = {On the Termination of Integer Loops},
year = {2012},
issue_date = {December 2012},
publisher = {Association for Computing Machinery},
address = {New York, NY, USA},
volume = {34},
number = {4},
issn = {0164-0925},
url = {https://doi.org/10.1145/2400676.2400679},
doi = {10.1145/2400676.2400679},
journal = {ACM Trans. Program. Lang. Syst.},
month = dec,
articleno = {16},
numpages = {24}
}

@misc{lagarias20213x1problemoverview,
      title={The 3x+1 Problem: An Overview}, 
      author={Jeffrey C. Lagarias},
      year={2021},
      eprint={2111.02635},
      archivePrefix={arXiv},
      primaryClass={math.NT},
      url={https://arxiv.org/abs/2111.02635}, 
}

@article{matthews2010generalized,
  title={Generalized 3x+ 1 mappings: markov chains and ergodic theory},
  author={Matthews, Keith R},
  journal={The ultimate challenge: The 3x},
  volume={1},
  pages={79--103},
  year={2010}
}

@inproceedings{Tiwari04:CAV,
	TITLE = {Termination of linear programs},
	AUTHOR = {Tiwari, A.},
	BOOKTITLE = {Computer-Aided Verification, CAV},
	EDITOR = "Alur, R. and Peled, D.",
	PAGES = {70--82},
	PUBLISHER = {Springer},
	SERIES = {LNCS},
	VOLUME = "3114",
	MONTH = jul,
	YEAR = 2004
}

@InProceedings{guilmant_et_al:LIPIcs.ICALP.2024.140,
  author =	{Guilmant, Quentin and Lefaucheux, Engel and Ouaknine, Jo\"{e}l and Worrell, James},
  title =	{{The 2-Dimensional Constraint Loop Problem Is Decidable}},
  booktitle =	{51st International Colloquium on Automata, Languages, and Programming (ICALP 2024)},
  pages =	{140:1--140:21},
  series =	{Leibniz International Proceedings in Informatics (LIPIcs)},
  ISBN =	{978-3-95977-322-5},
  ISSN =	{1868-8969},
  year =	{2024},
  volume =	{297},
  editor =	{Bringmann, Karl and Grohe, Martin and Puppis, Gabriele and Svensson, Ola},
  publisher =	{Schloss Dagstuhl -- Leibniz-Zentrum f{\"u}r Informatik},
  address =	{Dagstuhl, Germany},
  URL =		{https://drops.dagstuhl.de/entities/document/10.4230/LIPIcs.ICALP.2024.140},
  URN =		{urn:nbn:de:0030-drops-202831},
  doi =		{10.4230/LIPIcs.ICALP.2024.140}
}

@InProceedings{10.1007/11817963_34,
author="Braverman, Mark",
editor="Ball, Thomas
and Jones, Robert B.",
title="Termination of Integer Linear Programs",
booktitle="Computer Aided Verification",
year="2006",
publisher="Springer Berlin Heidelberg",
address="Berlin, Heidelberg",
pages="372--385",
isbn="978-3-540-37411-4"
}

@InProceedings{hosseini_et_al:LIPIcs.ICALP.2019.118,
  author =	{Hosseini, Mehran and Ouaknine, Jo\"{e}l and Worrell, James},
  title =	{{Termination of Linear Loops over the Integers}},
  booktitle =	{46th International Colloquium on Automata, Languages, and Programming (ICALP 2019)},
  pages =	{118:1--118:13},
  series =	{Leibniz International Proceedings in Informatics (LIPIcs)},
  ISBN =	{978-3-95977-109-2},
  ISSN =	{1868-8969},
  year =	{2019},
  volume =	{132},
  editor =	{Baier, Christel and Chatzigiannakis, Ioannis and Flocchini, Paola and Leonardi, Stefano},
  publisher =	{Schloss Dagstuhl -- Leibniz-Zentrum f{\"u}r Informatik},
  address =	{Dagstuhl, Germany},
  URL =		{https://drops.dagstuhl.de/entities/document/10.4230/LIPIcs.ICALP.2019.118},
  URN =		{urn:nbn:de:0030-drops-106940},
  doi =		{10.4230/LIPIcs.ICALP.2019.118}
}

@InProceedings{benamram:LIPIcs.STACS.2013.514,
  author =	{Ben-Amram, Amir M.},
  title =	{{Mortality of Iterated Piecewise Affine Functions over the Integers: Decidability and Complexity}},
  booktitle =	{30th International Symposium on Theoretical Aspects of Computer Science (STACS 2013)},
  pages =	{514--525},
  series =	{Leibniz International Proceedings in Informatics (LIPIcs)},
  ISBN =	{978-3-939897-50-7},
  ISSN =	{1868-8969},
  year =	{2013},
  volume =	{20},
  editor =	{Portier, Natacha and Wilke, Thomas},
  publisher =	{Schloss Dagstuhl -- Leibniz-Zentrum f{\"u}r Informatik},
  address =	{Dagstuhl, Germany},
  URL =		{https://drops.dagstuhl.de/entities/document/10.4230/LIPIcs.STACS.2013.514},
  URN =		{urn:nbn:de:0030-drops-39615},
  doi =		{10.4230/LIPIcs.STACS.2013.514}
}

@inproceedings{inproceedings,
author = {Podelski, Andreas and Rybalchenko, Andrey},
year = {2004},
month = {01},
pages = {239-251},
title = {A Complete Method for the Synthesis of Linear Ranking Functions},
booktitle    = {Verification, Model Checking, and Abstract Interpretation, 5th International
                  Conference, {VMCAI} 2004, Venice, Italy, January 11-13, 2004, Proceedings},
volume = {2937},
isbn = {978-3-540-20803-7},
doi = {10.1007/978-3-540-24622-0_20}
}

@misc{benamram2025terminationanalysislinearconstraintprograms,
      title={Termination Analysis of Linear-Constraint Programs}, 
      author={Amir M. Ben-Amram and Samir Genaim and Joël Ouaknine and James Worrell},
      year={2025},
      eprint={2509.06752},
      archivePrefix={arXiv},
      primaryClass={cs.PL},
      url={https://arxiv.org/abs/2509.06752}, 
}

@book{10.5555/17634,
author = {Schrijver, Alexander},
title = {Theory of linear and integer programming},
year = {1986},
isbn = {0471908541},
publisher = {John Wiley \& Sons, Inc.},
address = {USA}
}

@misc{telen2022introductiontoricgeometry,
      title={Introduction to Toric Geometry}, 
      author={Simon Telen},
      year={2022},
      eprint={2203.01690},
      archivePrefix={arXiv},
      primaryClass={math.AG},
      url={https://arxiv.org/abs/2203.01690}, 
}

@article{Lenstra1983IntegerPW,
  title={Integer Programming with a Fixed Number of Variables},
  author={Hendrik W. Lenstra},
  journal={Math. Oper. Res.},
  year={1983},
  volume={8},
  pages={538-548},
  url={https://api.semanticscholar.org/CorpusID:7868492}
}

@book{10.5555/1370949,
author = {Berg, Mark de and Cheong, Otfried and Kreveld, Marc van and Overmars, Mark},
title = {Computational Geometry: Algorithms and Applications},
year = {2008},
isbn = {3540779736},
publisher = {Springer-Verlag TELOS},
address = {Santa Clara, CA, USA},
edition = {3rd ed.}
}

@book{rockafellar2015convex,
  author       = {R. Tyrrell Rockafellar},
  title        = {Convex Analysis},
  series       = {Princeton Landmarks in Mathematics and Physics},
  publisher    = {Princeton University Press},
  year         = {1970},
  isbn         = {978-1-4008-7317-3},
  bibsource    = {dblp computer science bibliography, https://dblp.org}
}

@article{Matthews1985AMA,
  title={A Markov approach to the generalized Syracuse algorithm},
  author={Keith R. Matthews and A. M. Watts},
  journal={Acta Arithmetica},
  year={1985},
  volume={45},
  pages={29-42},
  url={https://api.semanticscholar.org/CorpusID:118600406}
}

@article{Leigh1986AMP,
  title={A Markov process underlying the generalized Syracuse algorithm},
  author={George M. Leigh},
  journal={Acta Arithmetica},
  year={1986},
  volume={46},
  pages={125-143},
  url={https://api.semanticscholar.org/CorpusID:117635633}
}

@article{Mller1978berHV,
  title={{\"U}ber Hasses Verallgemeinerung des Syracuse-Algorithmus (Kakutanis Problem)},
  author={H. M{\"o}ller},
  journal={Acta Arithmetica},
  year={1978},
  volume={34},
  pages={219-226},
  url={https://api.semanticscholar.org/CorpusID:117690327}
}

@article{Matthews1992SomeBM,
  title={Some Borel measures associated with the generalized Collatz mapping},
  author={Keith R. Matthews},
  journal={Colloquium Mathematicum},
  year={1992},
  volume={63},
  pages={191-202},
  url={https://api.semanticscholar.org/CorpusID:54607478}
}

@InProceedings{10.1007/978-3-319-02444-8_26,
author="Heizmann, Matthias
and Hoenicke, Jochen
and Leike, Jan
and Podelski, Andreas",
editor="Van Hung, Dang
and Ogawa, Mizuhito",
title="Linear Ranking for Linear Lasso Programs",
booktitle="Automated Technology for Verification and Analysis",
year="2013",
publisher="Springer International Publishing",
address="Cham",
pages="365--380",
isbn="978-3-319-02444-8"
}

@inproceedings{10.1007/11817963_37,
author = {Cook, Byron and Podelski, Andreas and Rybalchenko, Andrey},
title = {TERMINATOR: beyond safety},
year = {2006},
isbn = {354037406X},
publisher = {Springer-Verlag},
address = {Berlin, Heidelberg},
url = {https://doi.org/10.1007/11817963_37},
doi = {10.1007/11817963_37},
booktitle = {Proceedings of the 18th International Conference on Computer Aided Verification},
pages = {415–418},
numpages = {4},
location = {Seattle, WA},
series = {CAV'06}
}

@InProceedings{10.1007/978-3-642-39799-8_61,
author="Beyene, Tewodros A.
and Popeea, Corneliu
and Rybalchenko, Andrey",
editor="Sharygina, Natasha
and Veith, Helmut",
title="Solving Existentially Quantified Horn Clauses",
booktitle="Computer Aided Verification",
year="2013",
publisher="Springer Berlin Heidelberg",
address="Berlin, Heidelberg",
pages="869--882",
isbn="978-3-642-39799-8"
}

@InProceedings{10.1007/978-3-030-81685-8_35,
author="Unno, Hiroshi
and Terauchi, Tachio
and Koskinen, Eric",
editor="Silva, Alexandra
and Leino, K. Rustan M.",
title="Constraint-Based Relational Verification",
booktitle="Computer Aided Verification",
year="2021",
publisher="Springer International Publishing",
address="Cham",
pages="742--766",
isbn="978-3-030-81685-8"
}

@INPROCEEDINGS{1319598,
  author={Podelski, A. and Rybalchenko, A.},
  booktitle={Proceedings of the 19th Annual IEEE Symposium on Logic in Computer Science, 2004.}, 
  title={Transition invariants}, 
  year={2004},
  volume={},
  number={},
  pages={32-41},
  doi={10.1109/LICS.2004.1319598}}

@article{10.1145/322276.322287,
author = {Papadimitriou, Christos H.},
title = {On the complexity of integer programming},
year = {1981},
issue_date = {Oct. 1981},
publisher = {Association for Computing Machinery},
address = {New York, NY, USA},
volume = {28},
number = {4},
issn = {0004-5411},
url = {https://doi.org/10.1145/322276.322287},
doi = {10.1145/322276.322287},
journal = {J. ACM},
month = oct,
pages = {765–768},
numpages = {4}
}

\end{document}